%% file: main.tex
\documentclass[sigplan,screen,nonacm]{acmart}
\usepackage{cleveref,thm-restate}
\usepackage{stmaryrd}
\usepackage{proof}
\usepackage{graphicx}
\usepackage{xypic}
\usepackage{tikz}
\usetikzlibrary{calc}
\usepackage{enumerate}   
\usepackage{multicol}

\usepackage[disable
]{todonotes}

\usepackage{algorithm}
\usepackage[noend]{algpseudocode}

\usepackage{url}
\usepackage{mathtools}
\usepackage{color}
\usepackage{cmll}
\usepackage{wrapfig}

\usepackage{mathtools}
\usepackage{xspace}
\Crefname{equation}{Eq.}{Eqs.}
\Crefname{figure}{Fig.}{Figs.}
\Crefname{theorem}{Thm.}{Thm.}
\Crefname{thm}{Thm.}{Thm.}
\Crefname{proposition}{Prop.}{Prop.}
\Crefname{prop}{Prop.}{Prop.}
\Crefname{remark}{Rem.}{Rem.}
\Crefname{section}{Sec.}{Sections}
\Crefname{corollary}{Cor.}{Cor.}
\Crefname{cor}{Cor.}{Cor.}
\Crefname{definition}{Def.}{Def.}

\newcommand{\version}{0}
\newcommand{\commenti}{0} 
\newcommand{\condinc}[2]{\ifthenelse{\equal{\commenti}{0}}{#1}{\blue{#2}} }
\newcommand{\SLV}[2]{\ifthenelse{\equal{\version}{0}}{#1}{ \RED{#2}}}

\input{macros}

\begin{document}


\title{ Bayesian Networks and Proof-Nets:\\ a proof-theoretical account  of Bayesian Inference}
\subtitle{
	{ \normalsize	Presented at the Workshop \textbf{\emph{Logic of Probabilistic Programming}} (CIRM Marseille, 31 Jan- 4 Feb 2022)  \url{https://conferences.cirm-math.fr/2507.html}}
}

\author{Thomas Ehrhard}
\author{Claudia Faggian}
\author{Michele Pagani}

\affiliation{%
	\institution{Université de Paris, CNRS, IRIF}
	\country{F-75013, Paris, France}
}


\begin{abstract}
We uncover a strong correspondence between Bayesian Networks and (Multiplicative) Linear Logic Proof-Nets, relating the two  as a representation of a joint  probability distribution and at the level of computation, so yielding  a proof-theoretical account of Bayesian Inference.

In particular, we relate the most widely used algorithm for exact inference (Message Passing over clique trees) with the inductive interpretation (in probabilistic coherence spaces) of a proof, opportunely factorized into a composition of smaller proofs. The correspondence turns out to be so tight, that even the computational cost of inference is similar.
\end{abstract}

%

\maketitle

\input{01_Preliminaries}


\input{02_MLLProofNets}

\input{03_PN_Inference}

\input{04_VE_bpn}

\input{04_VE_algo}
\input{07_Conclusions}

\bibliographystyle{amsplain}
\bibliography{biblioB}

\newpage
\appendix
\input{99_Appendix}
\input{99_Factorization}
\input{99_MessagePassing}

%
%

\end{document}

%% file: macros.tex
\def\eg{{\it e.g.}\xspace}
\def\ie{{\it i.e.}\xspace}

\newcommand\todoc[2][]{\todo[color=pink!20,#1]{#2}} 
\newcommand\todocf[1]{\todoc[inline]{#1}}

\newcommand{\pink}[1]{{\color{magenta}{#1}}}

\newcommand{\RED}[1]{{\color{red}{#1}}}
\newcommand{\blue}[1]{{\color{blue}{#1}}}

\newtheorem{theorem}{Theorem}
\newtheorem*{thm*}{Theorem}
\newtheorem*{prop*}{Proposition}
\newtheorem*{lemma*}{Lemma}
\newtheorem*{example*}{Example}
\newtheorem*{definition*}{Definition}
\newtheorem{definition}[theorem]{Definition}
\newtheorem{lemma}[theorem]{Lemma}

\newtheorem{corollary}[theorem]{Corollary}
\newtheorem{proposition}[theorem]{Proposition}
\newtheorem{example}[theorem]{Example}
\newtheorem{remark}[theorem]{Remark}
\newtheorem{notation}[theorem]{Notation}

\newtheorem{prop}[theorem]{Proposition}
\newtheorem{cor}[theorem]{Corollary}

\newtheorem{property}[theorem]{Property}

\newcommand{\textdef}[1]{\emph{#1}} 

\newcommand{\BigO}[1]{\mathbf{O}\left(#1\right)}
\newcommand{\Exp}[1]{\mathsf{exp}\left(#1\right)}

\newcommand{\st}{\mathrel{|}}

\newcommand{\Red}[1][]{\xrightarrow{#1}}	

\newcommand{\MLL}{\ensuremath{\mathsf{MLL}}\xspace} 
\newcommand{\One}{\mathsf 1} 
\newcommand{\Bottom}{\mathsf \perp} 

\newcommand{\FormA}{F} 
\newcommand{\FormB}{G} 

\newcommand{\RP}{\ensuremath{\mathbb{R}_{\geq 0}}} 

\newcommand{\VecA}{\alpha} 

\newcommand{\MatA}{\phi} 

\newcommand{\Ev}{\iota} 

\newcommand{\PCoh}{\ensuremath{\mathsf{PCoh}}} 
\newcommand{\Pcs}[1]{\ensuremath{\mathcal{#1}}} 
\newcommand{\Clique}[1]{\ensuremath{\mathsf{P}(#1)}} 
\newcommand{\Web}[1]{\vert #1 \vert} 
\newcommand{\App}[2]{#1\cdot #2} 

\newcommand{\Sem}[2][]{\llbracket #2\rrbracket^{#1}} 

\newcommand{\BigCProd}{\prod} 

\newcommand{\Proj}[3][]{#2\vert_{#3}^{#1}} 

\newcommand{\Val}[1]{\Web{#1}} 
\newcommand{\ValB}{y} 

\newcommand{\vX}{X} 

\newcommand{\FProd}{\odot} 
\newcommand{\BigFProd}{\bigodot} 
\newcommand{\FProj}[2]{\mathrm{proj}_{#2}\left(#1\right)} 

\newcommand{\tst}{\mathit{ s.t. }}

\newcommand{\sem}[1]{\llbracket {#1} \rrbracket}

\RequirePackage{ifthen}

\newcommand{\BN}{Bayesian Network\xspace}

\newcommand{\netB}{{\boldsymbol{\mathcal B}}}
\newcommand{\bnet}[1]{\mathscr{B}(#1)}
\newcommand{\Bnet}[1]{\netB(#1)}

\newcommand{\netN}{\mathcal N}
\newcommand{\N}{\netN}
\newcommand{\netR}{\mathcal R}
\newcommand{\R}{\netR}
\newcommand{\M}{\mathcal M}
\newcommand{\netD}{\mathcal D}
\newcommand{\D}{\netD}
\newcommand{\netG}{\mathcal G}
\newcommand{\netC}{\mathcal C}

\newcommand{\tree}{\mathcal T}

\newcommand{\CutS}[1]{\mathrm{Cut}(#1)}

\newcommand{\Var}[1]{\mathsf{Var}(#1)}
\newcommand{\scope}[1]{\Var{#1}}

\newcommand{\size}[1]{\#(#1)}

\newcommand{\bX}{\mathbb X}
\newcommand{\bY}{\mathbb Y}
\newcommand{\bZ}{\mathbb Z}

\newcommand{\bC}{\mathbb C}

\renewcommand{\vec}{\bar}

\newcommand{\bx}{{\vec x}}
\newcommand{\by}{{\vec y}}
\newcommand{\bz}{{\vec z}}

\newcommand{\cut}{\mathsf{cut}}
\newcommand{\Ax}{\mathsf{Box}}

\newcommand{\bn}{\mathsf{box}}
\newcommand{\ax}{\mathsf{ax}}
\newcommand{\w}{\mathsf{w}}
\newcommand{\weak}{\mathsf{w}}
\newcommand{\cn}{@}

\newcommand{\axrule}{\mathit{ax}}

\newcommand{\assoc}{\mathsf {c.ass}}
\newcommand{\neut}{\mathsf {c.w}}
\newcommand{\cid}{\mathsf {c.id}}

\newcommand{\project}{\FProj}

\newcommand*{\defeq}{\stackrel{\text{def}}{=}}

\newcommand{\rvar}{r.v.\xspace}
\newcommand{\rvars}{r.v.s\xspace}
\newcommand{\nax}[1]{#1^{\neg\Ax}}

\newcommand{\MLLB}{\ensuremath{\mathsf{MLL^{\mathtt B}}}\xspace}
\newcommand{\MLLAx}{\MLLB}
\newcommand{\factorized}{factorized\xspace}

\newcommand{\repfree}{repetition-free\xspace}
\newcommand{\atom}{atom\xspace}
\newcommand{\atoms}{atoms\xspace}

\newcommand{\At}[1]{\mathsf{Var}(#1)}

\newcommand{\CC}{CC\xspace}

\newcommand{\pair}[2]{\langle #1,#2\rangle}

\newcommand{\Nets}{\mathsf{Nets}}
\newcommand{\T}{\mathcal T}
\newcommand{\contree}{@-unit \xspace}

\newcommand{\InTrees}[1]{\mathsf{Nets}_{#1}}

\newcommand{\cutnet}{cut-net\xspace}

\newcommand{\monowired}{mono-wired net\xspace}

\newcommand{\Width}[1]{\mathrm{Width}(#1)}
\newcommand{\pol}[1]{\mathit{Pol}({#1})}

\newcommand{\iphi}{\iota}

\newcommand{\cC}{\mathbf{C}}
\newcommand{\cS}{\mathbf{S}}

\newcommand{\Xn}{X^-}
\newcommand{\Yn}{Y^-}
\newcommand{\Zn}{Z^-}
\newcommand{\X}{X^+}
\newcommand{\Y}{Y^+}
\newcommand{\Z}{Z^+}

\newcommand{\supported}{supported\xspace}
\newcommand{\Xcirc}{X^\circ}

\newcommand{\good}{well-labelled\xspace}
\newcommand{\goodness}{good labelling\xspace}

\newcommand{\Bpn}{Bayesian proof-net\xspace}

\newcommand{\pn}{proof-net\xspace}

\newcommand{\bpn}{\texttt{bpn}\xspace}
\newcommand{\bpns}{\texttt{bpn's}\xspace}
\newcommand{\phisem}[1]{\sem {#1}^{\iota}}
\newcommand{\Wirings}[1]{\mathsf{Wirings}(#1)}
\newcommand{\wiring}{wiring-net\xspace}
\newcommand{\wirings}{wiring-nets\xspace}

\newcommand{\wiringM}{wiring-module\xspace}

\newcommand{\bbox}{\mathtt{b}}

\def\ih{{\it i.h.}\xspace}

\newcommand{\Hide}[2]{\mathsf{Hide}(#2 / #1)}
\newcommand{\Show}[2]{\mathsf{Show}(#2 /  #1)}
\newcommand{\W}[2]{\mathsf{Hide}(#2 / #1)}
\newcommand{\deW}[2]{\mathsf{Show}(#2 /  #1)}

\newcommand{\El}[2]{\mathsf{Factorize}(#1,#2)}

\newcommand{\groundinterpretation}{valuation\xspace}
\mathchardef\mhyphen="2D 
\renewcommand{\Box}{\Ax}
\newcommand{\EE}{\mathbf E}
\newcommand{\NN}{\mathbf N}

\newcommand{\Cliques}[2]{\mathsf{Cliques}(#1,#2)}


%% file: 01_Preliminaries.tex

\newcommand{\true}{\texttt{t}}
\newcommand{\false}{\texttt{f}}
\newcommand{\CPT}{CPT\xspace}
\newcommand{\aphi}{\hat  \phi}

\section{Introduction}

Bayesian Networks \cite{Pearl88} are a prominent  tool for probabilistic reasoning, allowing for a compact (factorized) representation of large probability distributions, and for efficient inference algorithms. 
A  Bayesian Network consists of two parts: a qualitative component, given by 
a directed acyclic graph, and a quantitative component, given by conditional  probabilities.
This  bears a striking resemblance with proof-nets of Linear Logic \cite{ll,synsem}: proof-nets   are  a  graph  representation of the syntax and dynamics  (cut-elimination) of proofs,  to which can be associated a quantitative interpretation.  The goal of this paper is to 
give a clear status to this correspondence,  yielding to a proof-theoretical account of Bayesian inference.

Linear Logic  has brought several  insights 
 into proof theory and (via  the Curry-Howard correspondence between proofs and programs ) into the semantics of  programming language. Such insights have  proved especially  suitable for modelling 
 probabilistic programming. In particular, Linear Logic has 
  enabled an interactive view,   expressing the flow of computation (Geometry of Interaction, \eg  \cite{popl17, LagoH19},    Game Semantics, \eg \cite{DanosH02,CastellanCPW18}), and the account for resources, leading in particular  to quantitative semantics based on linear algebra operations, such as  {Probabilistic Coherence Space semantics} (PCoh)\cite{Girard2003, danosehrhard,EhrPagTas14,EhahrdPT18fa,EhrhardT19}.

\vspace{-3pt}
\paragraph{Factorized representation and factorized computation.}

A joint distribution is a \emph{global} function involving many variables. 
A common way to deal with such a complex function is to  factorize  it  as a product of \emph{local} functions, each of  which depends on a subset of the variables.  Factorization  may involve both the computation  and the representation.
For example, a  joint probability  can be factorized as the product of  conditional probabilities.
Such an   approach  underlies  Bayesian Network, and inference algorithms  such as \emph{message passing}.
\emph{Message passing} has in fact  been (re)discovered and  applied, in various forms, in several domains (we refer to \cite{KschischangFL01} for a survey).

\vspace{-3pt}
\paragraph{Contents.}\Cref{sec:pre} recalls standard material. 
In  \Cref{sec:bpn}, we define  a sub-class  of   multiplicative  proof-nets (we refer to them as Bayesian proof-nets)
which we then  show to correspond to Bayesian Networks. 
Under opportune conditions, the \emph{PCoh denotation} of a Bayesian proof-net  $\R$ represents the same joint probabillity distribution as the  \BN    $\Bnet{\R}$ which is   
associated to $\R$ (\Cref{sec:PR}). Furthermore, $\Bnet{\R}$  turns out to be an \emph{invariant of the proof-net} under normalization and other relevant transformations.

 In sections \ref{sec:factorized_form} to \ref{sec:induced},  we \emph{factorize} a proof-net in the composition of smaller nets, whose interpretation has a smaller cost. 
In proof-theory, the natural way to \emph{factorize  a proof} in smaller components is to factorize  it in sub-proofs which are then \emph{composed via cut}. We follow exactly this  way.	We then show that---remarkably---there is a tight correspondence between the \emph{decomposition of  a proof-net},
and the well-known decomposition of Bayesian Networks into \emph{clique trees}. The inductive \emph{PCoh interpretation} of a factorized proof-net closely corresponds to \emph{message passing over clique trees} (\Cref{sec:MP}).
The correspondence turns out to be so tight, that the respective computational costs are similar.


\condinc{}{
\blue{	Bayesian Networks are a graphical  modeling tool for  compactly specifying joint probability distributions, and 
	to facilitate efficient inference. 
	
	%

	Graph-based inference 
	Compilation into a graphical model such as BN 
	is at the core of inference in 
	probabilistic programming languages such as Anglican have    \BN  at the core of their inference 
	
	A Bayesian network consists of two parts: a qualitative component in the form of
	a directed acyclic graph (DAG), and a quantitative component in the form of  conditional
	probabilities. 
	This has a striking similarity with proof-nets of 
	Linear Logic
	
	Curry–Howard correspondence establish 
	direct relationship between computer programs and mathematical proofs.  a proof is a program, and the formula it proves is the type for the program)  which involves not only the representation (proof-as-program, formula-as types) but also the computation (execution of a program vs cut-elimination)
	}
}

\condinc{}{
	An appraoch to inference taken by probabilistic programs such as Infer.Net  or Anglican 
	is to compile  to a  a probabilistic model such as a Bayesian network
}

\vspace{-3pt}
\paragraph{Related work}
The relationship between  Bayesian Networks and logic  has   been  put forward in the setting of categorical semantics  \cite{JacobsZ}.
The resemblance between Bayesian Networks and proof-nets is not a new observation, but to our knowledge has never been formally developed.  
A key ingredient  to bridge between the two world is the quantitative interpretation offered by  Probabilistic Coherence Spaces. A sibling paper develops the formulation of PCoh in terms of factors  in the  setting of multiplicative Linear Logic  (here, we restrict our attention to the class of  Bayesian proof-nets). The goal there is to then  use  inference algorithms to efficiently compute the denotation of any proof.
In this paper, we follow a different directions: we  focus on the proof-nets, and on inference-as-interpretation where the data structure which supports the computations is the proof-net itself. 
%

\subsection{Motivating Examples}
\subsubsection{An example of Bayesian Inference}
\begin{figure}
		\includegraphics[page=2,width=0.9\linewidth]{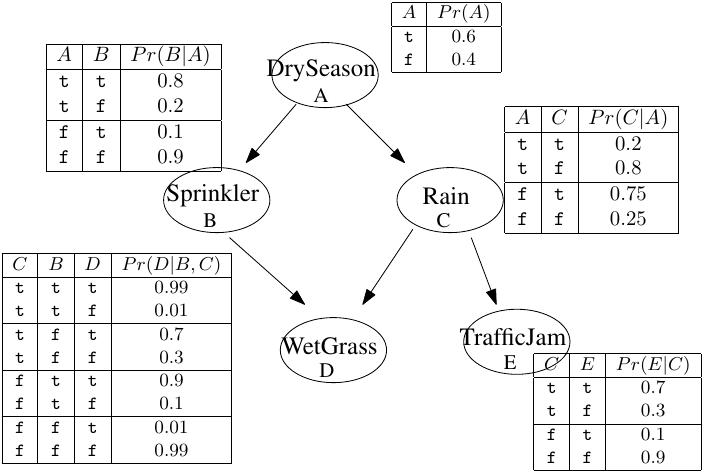}
	\caption{}
	\label{fig:BNrain}
\end{figure}
Let us start with an informal example. We want
to model  the  fact that the lawn being  Wet in the morning may depends on either Rain or the Sprinkler being on. In turn, both Rain and the regulation of the  Sprinkler  depend on the Season. Moreover, Traffic Jams are correlated with Rain.   
 The dependencies between these  five variables (shorten  into $ A,B,C,D,E $) are pictured in  \Cref{fig:BNrain},  where 
  the strength of the dependencies is quantified by 
   conditional probability tables.
Assume we wonder: did it rain last night? 
Assuming we are in DrySeason, our \emph{prior} belief is that Rain happens with probability $ 0.2 $. However, if we observe that the lawn is Wet, our confidence will increase.  The updated belief is called \emph{posterior}.
The model in \Cref{fig:BNrain} allows us to infer the \emph{posterior probability} of Rain, given the   evidence, formally:  $\Pr(R=\true | W=\true)$, or  to 
infer  how likely is it that the lawn is wet, \ie infer the \emph{marginal}  $\Pr(W=\true)$. 
Conditional probabilities and marginals are typical queries which can be answered by \emph{Bayesian inference}, whose core is Bayes conditioning:

\vspace{-8pt}
{{\small \begin{align*}
			\Pr(X=x|Y=y)=  \dfrac{\Pr(X=x,Y=y)}{\Pr(Y=y)}   = \dfrac{\Pr(Y=y|X=x) \Pr(X=x)}{\Pr(Y=y)} 
\end{align*}}}
\condinc{}{\begin{align}
	\Pr(X=x|Y=y)&=  \dfrac{\Pr(X=x,Y=y)}{\Pr(Y=y)} \label{CP} \\& = \dfrac{\Pr(Y=y|X=x) \Pr(X=x)}{\Pr(Y=y)}\label{BC}
\end{align} 

The   numerator of  the equations above often suffices, since the  posterior is proportional to it: 
\[\text{Posterior} \propto \text{Likehood }\cdot \text{Prior}\]
}
So, to compute the posterior $\Pr(\mathsf{Rain}=\true|\mathsf{Wet}=\true)$, we have to  compute the marginal $\Pr(\mathsf{Rain}=\true,\mathsf{Wet}=\true)$, which can be  obtained by 
summing out the other variables  from     the joint probability (\emph{marginalization}). 
The  marginal probability  $ \Pr(\mathsf{Wet}=\true)  $ of the evidence  is   computed in a similar way.
We then obtain the posterior by normalizing.
Notice that further evidence (for example, information on TrafficJams ) would  once again update our belief.

\emph{Summin Up.} The \BN in \Cref{fig:BNrain} expresses the joint probability 
of the five  variables, in a way which is  compact  (\emph{factorized representation}) and which allows for efficient inference (\emph{factorized computation}).

\condinc{}{
\[	\begin{array}{|c|c|c|}
	\hline
D&	R	&Pr(R|D)  \\
	\hline
\true&	\true & 0.2 \\
\true&	\false	& 0.8\\
	\hline
\false	&	\true & 0.75 \\
\false	&	\false	& 0.25\\
\hline	
\end{array}
\quad
\begin{array}{|c|c|c|}
	\hline
	D&	R	&Pr(R|D)  \\
	\hline
	\true&	\true & 0.8 \\
	\true&	\false	& 0.2\\
	\hline
	\false	&	\true & 0.1 \\
	\false	&	\false	& 0.9\\
	\hline	
\end{array}
\quad
\begin{array}{|c|}
	\hline
	\dots\\
	\hline	
\end{array}
\begin{array}{|c|c|c|}
	\hline
	R&	T	&Pr(T|R)  \\
	\hline
	\true&	\true & 0.7 \\
	\true&	\false	& 0.3\\
	\hline
	\false	&	\true & 0.1 \\
	\false	&	\false	& 0.9\\
	\hline	
\end{array}
\]
} 

\condinc{}{
{\ \[\begin{array}{c  c c}
	\begin{array}{|c|c|}
			\hline
		R	&Pr(R)  \\
		\hline
		\true & 0.2 \\
		\false	& 0.8\\
			\hline
	\end{array}
\quad
\begin{array}{|c|c|}
		\hline
	S	&  Pr(S) \\
	\hline
	\true & 0.1 \\
	\false	& 0.9\\
		\hline
\end{array}
\quad 
\begin{array}{|c|c|c|c|}
	\hline
R	& S  &  W& Pr(W|R,S)\\
\hline 
\true	& \true & \true & 0.99 \\
\true	& \true & \false & 0.01 \\
\hline
\true	&  \false & \true  & 0.7\\
\true	&  \false & \false  & 0.3\\
\hline
\false	& \true &  \true& 0.9  \\
\false	& \true &  \true& 0.1  \\
\hline
\false	& \false & \true & 0.01\\
\false	& \false & \false & 0.99\\
\hline
\end{array}
\end{array}
\]
}

}


\subsubsection{An example of Proof-Net}
\begin{figure*}
	\centering
	\begin{minipage}[b]{7cm}
		\includegraphics[page=9,width=1\linewidth]{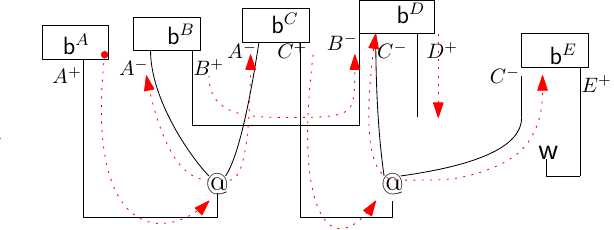}
		\caption{ Proof-net $\R_D$}
		\label{fig:ABCDE_D}
	\end{minipage}
	\quad
	\begin{minipage}[b]{9cm}
		\scriptsize{ $\infer[\cut (A)] {\quad\vdash D^+}
			{\infer[\bbox^A ]{\vdash A^+}{}&
				\infer[\cut(C) +@]{\vdash A^-,D^+ }
				{\infer[\bbox^C]{\vdash A^-,C^+}{} & \infer[\cut (B)]{\vdash A^-,C^-,D^+}{\infer[\bbox^B]{\vdash A^-, B^+}{} & 
						{
							\infer[cut(E)]{\vdash B^-,C^-,D^+}{\infer[mix+@]{\vdash B^-,C^-,D^+,E^+}
								{\infer[\bbox^D]{\vdash B^-,C^-,D^+}{} & \infer[\bbox^E]{\vdash C^-,E^+}{} }  
								& \infer[\w]{\vdash E^-}{} }
						}
					}
				}
			}
			$}
		\caption{A sequentialization of the proof-net $\R_D$}
		\label{fig:ABCDE_seq}
	\end{minipage}
\end{figure*}

Proof-nets \cite{ll,synsem}, are a graphical representation of  Linear Logic sequent calculus proofs.
The  \BN in \Cref{fig:BNrain} can be encoded in \emph{multiplicative}  linerar logic (\MLL) as the  proof-net in  \Cref{fig:ABCDE_D}.   A sequent calculus proof which corresponds to this \pn is in \Cref{fig:ABCDE_seq}.
The nodes $\bbox^A, \bbox^B, \bbox^C, ...$ (called \emph{boxes}, and corresponding to generalized axioms)
are  containers in which to store semantical information---for example,  the same conditional probability distributions as in \Cref{fig:BNrain}.
The purely \MLL part of the \pn acts as a \emph{wiring} which plugs together and transfers  such information. 

Notice that edges in \Cref{fig:ABCDE_D} are labelled by atomic formulas, either positive or negative.
The flow of information in a proof-net (its geometry of interaction \cite{goi0,goi1,multiplicatives}) follows the polarity of atoms, going downwards on positive atoms (which carry "output information")
and upwards on negative atoms  ("input information").
\MLL allows for \emph{duplication of the information carried by negative atoms}---via the   $\cn$-node (contraction). The $\w$-nodes (weakening) block 
the information.

The \pn 
$\R_D$ in \Cref{fig:ABCDE_seq} has a single conclusion $D^+$.
We will see that $\R_D$ represents a marginal probability $\Pr(D)$. 
Let us informally see how we can draw a sample from $\R_D$.
\begin{example}[Sampling from a \pn]
The only  node which is \emph{initial} (w.r.t the flow)  in this \pn is $\bbox^A$, which   outputs a sample from $\Pr(A)$. 	
Assume this value is $\true$. This sample is \emph{propagated} via the $\cn$-node (\emph{contraction})  to both $\bbox^C$ and $\bbox^B$. When $\bbox^C$  receives $\true$,  it samples a value from the distribution 
$\Pr(C|A=\true)$ (for example, $\true$ with probability $0.2$).   Assume the output of $\bbox^C$ is    $\true$, and that the output of   $\bbox^B$ (obtained with a similar procedure) is $\false$. When the box  $\bbox^D$
receives these values, it  outputs a sample from $\Pr(D|C=\true,B=\false)$, which is a sample from the marginal $\Pr(D)$.

\condinc{}{Notice  that all non initial boxes have  to wait for all their inputs (a form of synchronization \cite{popl17}).}
\end{example}

\condinc{}{
	We are now  able to informally see how we can draw a sample from a \pn---so clarifying the role of both \emph{weakening} and \emph{contraction}. 
	
\begin{example}[Sampling from a \pn] Let us consider the \pn  of single conclusion $D$
	in \Cref{fig:ABCDE2}. Assume that each box stores a conditional probability  (take the same CPT's as in \Cref{fig:BNrain}).
	The way information flows in the graph is captured by the 
	polarized order  (as standard in  Linear Logic).
	Nodes which are not initial  need to wait for inputs.
	
	The only initial node  in this \pn is $\bbox^A$, which  outputs  a sample from $\Pr(A)$. 	
	Assume this value is $\true$. This sample is \emph{propagated} via the $\cn$-node (\emph{contraction})  to both $\bbox^C$ and $\bbox^B$. When $\bbox^C$  receives $\true$,  it samples from
	$\Pr(C|B=\true)$ (which according to the CPT in \Cref{fig:BNrain} is $\true$ with probability $0.2$).  Assume the output of $\bbox^C$ is    $\true$, and that the output of   $\bbox^B$ (obtained with a similar procedure) is $\false$. When the box  $\bbox^D$
	receives these values (notice it needs to wait for both inputs), it  outputs a sample from $\Pr(D|C=true,D=false)$, which is a sample from the marginal $\Pr(D)$.
\end{example}
}

\section{Preliminaries}\label{sec:pre}
\subsection{ Bayesian modeling}\label{sec:basics}


\condinc{}{
An instantiation is a particular assignment of values to all variables. The configuration
space  is the set of all instantiations; it is the domain of the global function f.

In every fixed instantiation , every variable has some definite value. We may
therefore consider also the variables in a factor graph as functions with domain Ω. Mimicking the standard notation for random variables, we will denote such functions by capital
letters

Throughout this paper we deal with functions of many variables. Let , be a collection of variables, in which,
for each , takes on values in some (usually finite) domain
(or alphabet) . Let be an -valued function
of these variables, i.e., a function with domain

====
Inference over BNs can be either exact or approximate.
Perhaps the most popular exact inference algorithm,
belief propagation in junction trees, relies on the compilation of a BN into a junction tree. Exact belief updating (or marginalization) is then performed by message passing over the junction tree 

At a high-level the junction tree algorithm partitions the graph into clusters of variables; internally, the variables within a cluster could be highly coupled; however, interactions among clusters will have a tree structure,
}


\condinc{}{
=========

Given a probability space, a random variable $X$ is a
function  which  assigns a  value to each element (outcome)
in the sample space. The set of values random variable X can assume is called
the \emph{space} of $X$. A random variable is said to be \emph{discrete} if its space is finite
or countable. 

When doing Bayesian inference, there is some entity which has features,
the states of which we wish to determine, but which we cannot determine for
certain. So we settle for determining how likely it is that a particular feature is
in a particular state. 

A random variable represents some feature of the system 
being modeled, and we are uncertain as to the values of this features. 

For the purpose of probabilistic modeling,  a random variable $ X $ is  a symbol representing any
one of a set of values $\Val X$. 
==
}


Bayesian methods provide a formalism for reasoning about partial beliefs under
conditions of uncertainty. We  cannot determine for
certain  the state of some features of interest, and we settle for determining how \emph{likely} it is that a particular feature is in a particular state. 
\condinc{}{Think of    a system which has features,
the states of which we wish to determine, but which we cannot determine for
certain---so we settle for determining how \emph{likely} it is that a particular feature is
in a particular state. }
Random variables represent   features  of the system 
being modeled.   
For the purpose of probabilistic modeling,  a random variable  can be seen as  a symbol ("Wet") which assumes values from a space  ($\{\true,\false\}$).  The system  is  modeled as a joint probability on all possible values of the variables of interest -- an element in the sample space represents a possible state of the system. 

Notice that in Bayesian modeling, random variables are identified first, and only implicitly become functions on a sample space. 
We refer to the excellent textbook by Neapolitan \cite{NeapolitanBook} (Ch.~1) for a formal treatment relating the  notion of random variable as used in Bayesian inference, with the classical definition of function on a sample space.  

\subsubsection{Joint and Conditional Distributions}\label{sec:RV_basics}

 We  briefly revise the language  of Bayesian inference. We refer to \cite{DarwicheHandbook} for a concise presentation, and to standard texts for details \cite{Pearl88,DarwicheBook,NeapolitanBook}.
  Please notice that in this paper, we are only concerned with  random variables (\rvars, for short)  whose set of values is finite, so \emph{discrete random variables}.

\condinc{}
{
%

\begin{definition}[Set of random variables]\label{def:sets_rv}
	The metavariables $\bX,\bY,\bZ$  range over finite \textdef{sets of random variables}. Given  such a set $\bX$ we denote by 
	$\Val\bX$ the cartesian product $\BigCProd_{X\in\bX}\Val X$ of the value sets of the rv in $\bX$.

Given a subset $\bY\subseteq\bX$, we denote by $\Proj{\bx}{\bY}$ the restriction of $\bx$ to $\Val{\bY}$. 

	Given two sets of rv's $\bX$ and $\bY$, we say that $\bx\in\Val{\bX}$ and $\overline\ValB\in\Val{\bY}$ \textdef{agree on the common random variables}, $\bx\sim\overline\ValB$ for short, whenever $\Proj{\bx}{\bX\cap\bY}=\Proj{\overline\ValB}{\bX\cap\bY}$.
\end{definition}


%
%


}

To each   variable  $X$ is always associated a value set, which we denoted  $\Val X$.  Hence in this setting, by variable  $X$ we mean   a pair $(X, \Val X)$. 
 {A  set of variables $\bX=\{X_1, \dots, X_n\}$ defines a "compound" variable whose value set $\Val \bX$   is the Cartesian  product 
 	$ \BigCProd_{X\in\bX}\Val X$.}
 A \emph{value  assignment} for a set of variables $\bX$  is a mapping which  assigns to each $X\in \bX$ a value in $\Val X$.
 
	Given a  set of  variables $\bX$,
it is convenient and standard  
to assume as \emph{canonical sample space}\footnote{Alternatively, we can assume $\Omega$ to be the set of all value assignments for $\bX$}  the set $\Omega=\Val \bX$.
We can then see each  variable $X_i$ as a function from $\Omega$ to $\Val {X_i}$, projecting $ (x_1, \dots, x_n) $ into $x_i$.
Hence, a random variable in the classical sense.

\begin{notation}\label{not:agree}
	 By convention,  $x$ denotes a possible value for the random variable $ X$.
	%
\SLV{}{	Metavariables $\bX,\bY,\bZ$   range over finite \textdef{sets of random variables}.}
	The tuples in  $\Val\bX$ are denoted by overlined lowercase letters  such as $\bx, \bz$ .
	Each  element $\bx\in\Val\bX$ can be  seen as a value assignment (\ie, a mapping from $\bX$).
	 Given a subset $\bY\subseteq\bX$, we denote by $\Proj{\overline x}{\bY}$ the restriction of  $\bx$ to $\bY$. 
	%
	For two assignments $\bx$ (for $\bX$) and $\by$ (for  $\bY$), we write $\bx \sim \by$ 
	if they agree on the common variables, that is $\Proj{\bx}{\bX\cap\bY} =\Proj\by{\bX\cap\bY}$

\end{notation}

To each tuple  $\bx$ in the canonical sample space  $ \Val \bX $ is associated a probability, expressing a degree of belief in that values assignment. 	
\begin{definition}[Joint  probability distribution]
	  A  joint probability distribution $\Pr(\bX)$ over variables $\bX$ is a function from  $\Val \bX$  
	into the interval $[0,1]$ such that $\sum_{\bx\in \Val{\bX}} (\Pr(\bX))(\bx)=1$.

\end{definition}
\SLV{}{	An \emph{event} $\eta$ is a subset of $ \Val \bX $. 
$\Pr(\bX)$ assigns a
probability  to each event $\eta$ as 
$(\Pr(\bX))(\eta) = \sum_{\bx\in \eta} (\Pr(\bX))(\bx) $.
\SLV{}{Events are typically denoted by propositional sentences, which are defined inductively from atomic 
propositions having the form $X=x$.}
 Recall that $(X=x, Y=y)$ is short for the event $(X=x) \cap (Y=y)$.

\begin{example} The  canonical sample space for the five variables in \Cref{fig:BNrain}  consists of $2^5$ tuples. 
	The event $(R=\true) $  contains $2^4$ tuples,  the event $(R=\true) \cap (W=\true)$ is a subset of size $2^3$. 
\end{example}


\condinc{}{The notation $ \Pr (X | Z ) $ for conditional probability (\eqref{CP}) 
denotes a function. Intuitively, for each value of $Z$ , $ \Pr (X | Z ) $ gives a probability distribution over $X$.
}
A crucial notion when dealing with multiple variables is \textbf{conditional independence}: $X$ is independent from $Y$ given $Z$ if $\Pr(X|Y,Z) = \Pr(X|Y)$. 
Intuitively, learning the value of $Y$ does not provide additional information about $X$, once we know $Z$. This simplifies
  both the structure of probabilistic models, and the computations needed to perform inference, and is at the core of the definition of \BN.
\condinc{}{Conditional independence has a crucial role in simplifying both the structure of probabilistic models, and the computations needed to perform inference.
}

}

\vspace*{-4pt}
\subsubsection{Factorized, representation and  computation.}
The challenge of Bayesian  reasoning is that 
a joint probability distribution is usually too large to feasibly represent it explicitly. For example, a joint probability over $32$ binary  random variables, corresponds to $2^{32}$ entries.

\condinc{}
{The power of Bayesian networks stems  from both their
ability to represent large probability distributions compactly, and the availability of
inference algorithms that can answer queries about these distributions without necessarily constructing them explicitly. 
}
A \BN is a \emph{factorized representation} of an exponentially sized probability distribution. Inference algorithms then implement    \emph{ factorized  computations.}
\vspace*{-4pt}
\subsubsection{Bayesian Networks}\label{sec:BN}

A  Bayesian network (BN)  is a DAG in
which each node is a random variable. 
%
 \newcommand{\Pa}[1]{\mathsf{Pa}(#1)}
 \begin{definition}[ \BN over variables $\bX$]\label{def:BN}\label{def:bn}
 	A \BN $\netB$ over variables $\bX$ is a pair $(\mathcal G,\aphi)$ where
 	\begin{itemize}
 		\item $\mathcal G$ is a directed acyclic graph (DAG) over variables $\bX$, called  \emph{Bayesian structure}. 
 		$\Pa X$  denotes the parents of $X$ in $\mathcal G$. 
 		
 		\item $\aphi $   assigns, to each variable $\vX\in \bX$, a   \emph{conditional probability table} (\CPT) for $X$, noted $\phi^X$, which is interpreted as 
a conditional probability distribution, given   the parents of $X$ 
\begin{center}
		$ \phi^X=\Pr(X|\Pa X)	 $.
\end{center}

 	\end{itemize}
 	A $\BN$ over variables $\bX$ specifies   a \emph{unique} \cite{Pearl88} probability distributions
 	over $\bX$ (noted $\Pr_\netB$) ,  defined as  
 		\begin{align}\label{eq:BN}
	 			\Pr(\bX)(\bz) \defeq  \prod_{\substack{X\in \bX\\   x\by\sim \bz}} \phi^X(x\by) 
 		\end{align}
 where	$ x\by\sim \bz $ means that	 $  x\by $ agrees with  $ \bz $  on the values of their common variables (\Cref{not:agree}).
 	\end{definition}
 	The distribution in \Cref{eq:BN}  follows from \emph{a particular interpretation} of the structure and factors of $\netB=(\netG,\aphi)$ . Namely:
 	  each variable $X$ is assumed to be  \emph{independent} of its nondescendants  in $\mathcal G$, 
 	given its parents $\Pa X$
 	\footnote{A variable   $Z$ is  nondescendant of $X$ in the graph $\mathcal G$ if  $Z\not \in \{X\}\cup \Pa X$ and there is no directed path from 
 		$X$ to $Z$.};
 	 each  $\phi^X$ is interpreted as a conditional probability distribution.
 	
 Notice that 	a \CPT for $X$ is---essentially---just a stochastic matrix. It defines a probability distribution over $X$, given each possible assignment of values to its parents. 
 	\begin{definition}[\CPT]\label{def:CPT}
For $X\in \bX$, and $\bY\subseteq \bX$, 
a \CPT $\phi(X|\bY)$ for $X$ is a function (with domain  $\Val X \times \Val \bY$) mapping  
each tuple  $x\by$ 
to
a positive real   $\phi_{x\by}$
such that    $\sum_x \phi_{x\by} =1$. 

When $X,\bY$ are obvious, we simply  write $\phi$ for $\phi(X|\bY)$. 
 	\end{definition}

%


\condinc{}
{
 \begin{definition}[ \BN over variables $\bX$]\label{def:bn}
	A \BN $\netB$ over variables $\bX$ is a pair $(\mathcal G,\aphi)$ where
	\begin{itemize}
		\item $\mathcal G$ is a directed acyclic graph (DAG) over variables $\bX$. We denote   $\Pa X$  the parents of $X$ in $\mathcal G$. 
		
		\item $\aphi $   assigns to each variable $\vX\in \bX$ a factor $\phi^X$
		over  
		$\{X\}\cup \Pa X$ which  maps each  $x\by$  ($x\in \Val{X}$, $\by\in \Val{\Pa X}$) to
		a probability   $\phi_{x\by}$
		such that    $\sum_x \phi_{x\by} =1$.
	\end{itemize}
	A $\BN$ over variables $\bX$ defines  a \emph{unique} probability distributions
	over its variables,  defined as follows 
	\begin{align}\label{eq:BN}
		\Pr(\bX) \defeq  \BigFProd_{X\in \bX} \phi^X
	\end{align}
	\blue{
		\begin{align}\label{eq:BN}
			\Pr(\bX)(\bx) \defeq  \prod_{\substack{X\in \bX\\ \bz\in \Val{\Var{\phi^X}}~\tst~  \bz\sim \bx}} \phi^X(\bz)
		\end{align}
	}
\end{definition}

The distribution given by \Cref{eq:BN} follows from \emph{a particular interpretation} of the structure and factors of a \BN . Namely:
\begin{itemize}
	\item in $\mathcal G$,  each variable $X$ is assumed to be  independent of its nondescendants  $\bY$
	given its parents $\Pa X$\footnote{A variable   $Z$ is  nondescendant of $X$ in the graph $\mathcal G$ if  $Z\not \in \{X\}\cup \Pa(X)$ and there is no directed path from 
		$X$ to $Z$.}.
	\item each factor $\phi^X$ is interpreted as a conditional probability distribution 
	\[\phi^X=\Pr(X|\Pa X)\]
	
\end{itemize}
Such an interpretation is satisfied by an unique probability distribution \cite{Pearl88}, the one given by \Cref{eq:BN}.

In this sense  \Cref{eq:BN}  can also be stated as   \[\Pr(\bX) \defeq \BigFProd_{X\in \bX} \Pr(X|\Pa{X}).\]

}

\vspace*{-4pt}
\subsubsection{Factorized computation}
Inference algorithms  rely on basic operations on a class of functions 
(known as factors) which generalize   joint  probability distributions.
%
\begin{definition}[Factors]
	A \emph{factor} $\phi(\bX)$  over  variables $\bX$ is a function  from   $\Val{\bX}$  to $\mathbb R_{\geq 0}$
	(mapping each   $\bx\in \Val{\bX}$  to a 
	 real).
	 
	For readability (when $\bX$ is clear from the context) we   write    $\phi$ for $\phi(\bX)$.  $\Var{\phi}$  denotes $\bX$.
	We  often write $\phi_\bx$ for $\phi(\bx)$.
	
	Sum and product of factors are defined as follows.
	\begin{itemize}

				\item   Given  $\phi(\bX)$    and  $X\in \bX$, the result of \emph{summing out}  $X$ from  $\phi(\bX)$ is  a factor over  $ \bY = \bX  - \{X\} $ denoted by $\sum_X \phi $  or  $ \project{\phi}{\bY} $ (the {projection} of  $\phi$ over $\bY$), and  
				 defined as \begin{center}
				 	$(\sum_X \phi)_{\by}\defeq\sum\limits_{x\in \Val{X}}\phi_{x\by}$. 
				 \end{center}
		\item The result of \emph{multiplying} $\phi^1(\bX)$ and  $\phi^2(\bY)$  is a factor over  $\bZ = \bX\cup \bY$
		 denoted by  $\phi^1\FProd \phi^2$   and   defined as	
\begin{center}
	$ (\phi^1\FProd \phi^2)_\bz \defeq \phi^1_\bx \phi^2_\by  $
\end{center}
		where $\bx\sim \bz$ and $\by\sim \bz$  (\ie $\bx$ and $\by$ agree with $\bz$ on  the common variables).
	We denote   n-ary products by $\BigFProd_{n} \phi_n $.
		
	\end{itemize}
\end{definition}
The key property  used for inference in graphical models is that  
factor product and summation behave precisely as do product and summation over numbers. Both operations are commutative, product is associative, and---crucially---they distribute: if $X\not \in \Var{\phi} $, then 
$ \sum_X ( \phi_1 \FProd \phi_2)  = \phi_1 \FProd ( \sum_X  \phi_2 )  $, which allows for smaller computations.
%
%
\SLV{}{The  insight behind exact  inference algorithms is to 
exploit sum-product distributivity to decomposing the problem of inferring a marginal into smaller problems.
}

 Notice that CPT's are   factors.   \Cref{eq:BN}    can also be written:	
\begin{center}
	$ 	\Pr(\bX) \defeq  \BigFProd_{X\in \bX} \phi^X $
\end{center}
\begin{remark}[Cost of  operations on factors]\label{rem:factors_cost}
	Summing out any number of variables from a factor $\phi$ demands $\BigO{\Exp{w}}$ time and space, where $w$ is the number of variable over which $\phi$ is defined. To multiply $m$ factors demands $\BigO{m\cdot \Exp{w}}$ time and space, where $w$ is the number of variables in the resulting factor.
\end{remark}

%% file: 02_MLLProofNets.tex

\newcommand{\Module}{\M}
\newcommand{\contr}{\cn}
\newcommand{\ProofN}{\R}
\newcommand{\NodeB}{\bbox}
\newcommand{\AtM}{\mathscr{AT}}
	\newcommand{\arule}{\mathtt{a}}
\newcommand{\mrule}{\mathtt{m}}
\newcommand{\drule}{\mathtt{d}}
\newcommand{\Def}{:=} 

\newcommand{\OHide}{\mathsf{Hide}}
\newcommand{\OShow}{\mathsf{Show}}
\newcommand{\AtF}[1]{\mathtt{At}(#1)}
\condinc{}{\paragraph{Trees, Inputs and Output}
	
	Given a tree, the choice  of a  node $k$ as the root  determines an orientation of all edges  towards it.  Accordingly, we speak of \emph{inputs} and \emph{outputs} nodes and edges. 
	Notice that in a tree,  each non-root node   has \emph{exactly one output edge}.  
}

	\subsection{\MLLAx: Multiplicative Linear Logic  (with Boxes)}\label{sect:mll}

%
	We assume a countable alphabet of symbols, denoted by metavariables $ X,  Y,  Z$, to which we  refer as \emph{propositional variables, or atoms}.
	The grammar of the formulas is that  of the multiplicative fragment of linear logic (\MLL{} for short):
	\begin{align*}
		\FormA,\FormB&::= 
		X^+
		\mid X^-
		\mid\One
		\mid\Bottom
		\mid\FormA\otimes\FormB
		\mid\FormA\parr\FormB. 
	\end{align*}
	$ X^+$ (resp.~$ X^-$) denotes a \textdef{positive} (\emph{negative}) \emph{atomic formula}. 
	The \atom $X$ is the \emph{support} of the atomic formulas $\X,\Xn$. 
	
	Negation is defined by
		$( X^+)^\perp\Def X^-,( X^-)^\perp\Def X^+, 
	\One^\perp \Def\Bottom, \Bottom^\perp \Def\One,
	(\FormA\otimes\FormB)^\perp \Def\FormA^\perp\parr\FormB^\perp, (\FormA\parr\FormB)^\perp \Def\FormA^\perp\otimes\FormB^\perp.
	$


	A \textdef{sequent} is a finite sequence $\FormA_1,\dots, \FormA_n$ of \MLL{} formulas. Capital Greek letters $\Gamma,\Delta,\dots$  vary over sequents.  Given  $\Gamma=\FormA_1,\dots, \FormA_n$, we write $\Gamma^\perp$ for the sequent $\FormA_1^\perp,\dots, \FormA_n^\perp$. 

	In Linear Logic,  sequent calculus proofs are represented by special graphs. We first define modules, than proof-nets as a sub-class.
\condinc{}	{To  the standard $\MLL$-nodes we add a new sort of node, $\bn$.}
	\begin{definition}[pre-modules]
		\label{def:module}
		A \textdef{pre-module}\footnote{Usually called proof-structure in the literature.} $\Module$  is a labelled DAG with possibly pending edges (\ie some edges may have no source and/or no target)
		built over the alphabet of nodes which is
		represented in  \Cref{fig:pn} , where the  edges orientation is  top-bottom. The edges are labelled by \MLL formulas. 

		The \textdef{conclusions} (resp.~\textdef{premises}) of a node  are the outcoming (resp.~incoming) edges of that node. The \textdef{conclusions} (resp.~\textdef{premises}) of $\Module$ is the set of \emph{pending} edges of $\Module$ without target (resp.~source). 
		We  often write  
		"a conclusion  $ F $”  for “a conclusion
		 labelled by $F$”.

		$\Module$ is \textdef{atomic}  
		if  all formulas labelling its edges are atomic  
		formulas.   
		Metavariables $e,f$
		denote  edges.
		An edge $e$ is \emph{supported by} the atom $X$ ($e$ is positive/negative) if  its label is.
	\end{definition}
	The grammar 
		 in \Cref{fig:pn}  adds to  standard  $\MLL$ nodes, a new one $\bn$, which has \emph{atomic conclusions,  exactly one being positive. }   
\begin{notation}[ \MLL and \MLLAx] If we need to be specific, we  denoted  the extended calculus by \MLLAx.
 We  refer to  a (\emph{purely} multiplicative) module without $\bn$-nodes as    "$\MLL$ module".
\end{notation}
	\begin{remark}[Weakening and Contraction]\label{rk:negatives}
		The  syntax allows structural rules (weakening and contraction) on \emph{negative} atomic formulas. It is well-known in liner logic that $\bot$ admits weakening and contraction, and so do all the formulas built from $\bot$ by means of $\with $   and $\parr$ (\ie all formulas of \emph{negative polarity.})
		Recall that   \emph{a propositional variable stands for un unknown proposition }. Here, $\Xn$ stands for a  formula  of negative  polarity,
		and   $X^+$ for its dual\SLV{}{\footnote{For the purpose of this paper, the reader familiar with linear logic can simply think of $\Xn,\Yn,\Zn$ as occurrences of the same negative formula $\bot \with \bot$, and 
			$\X,\Y,\Z$ as occurrences of $\One\oplus \One$.}}.
	\end{remark}
	A proof-net is a pre-module  which fulfills a correctness criterion defined by means of switching paths (see [15]).
	\begin{definition}[Proof-nets]\label{def:pn}
		Given a \emph{raw} module $\Module$, a \textdef{switching path} in $\Module$ is an undirected path such that 
		for each $\parr$-node and $\cn$-node, the path uses  at most one    premise.
		
		$\Module$ is  a  \textdef{switching acyclic} module---or simply a \textbf{module}---
		whenever it has no switching cycle. 
		A \textbf{{proof-net}}  is a switching acyclic module which  has no pending premise. 
	\end{definition}

	The  conditions in \Cref{def:pn}  characterize the graphs which correspond to a sequent calculus proof. The sequent  calculus 
which corresponds to  \MLLAx proof-nets is standard---for  completeness we give  in \Cref{app:sequent} the generating rules and their image as \pn.
Please note that the calculus  includes a (non necessary, but convenient) \emph{mix}-rule.
\begin{theorem}[Sequentialization]\label{th:sequentialisation}
	Let $\Module$ be a \emph{pre-module} without  premises.  $\Module$ is a proof-net if, and only if, 
	$\Module$ 
	is the image of a sequent calculus derivation.
\end{theorem}
\begin{example}
	The derivation in \Cref{fig:ABCDE_seq} is a sequentialization of the proof-net in $\Cref{fig:ABCDE_D}$.
	
	\SLV{}{	Notice that (as rather standard) we  sequentialize in a proof where contractions are greedy, that is they are performed as soon as possible---so in fact we  embed contraction into the rules for $\cut,\otimes$, and mix.}
\end{example}

	\SLV{}{\begin{remark}\label{rk:subnet}
			Notice that  a sub-graph of a switching acyclic module is  switching acyclic. 
		\end{remark}
	}

			\begin{notation}[Boxes]\label{not:boxes}		 We call   \emph{box} (resp. axiom) a $\bn$-node ($\ax$-node) together with its conclusions.
		We   denote
		by $\Ax(\R)$  the \emph{set of all boxes} of $\R$. We denote by   $\nax \R$  the \MLL module obtained from $\R$ by removing all the $\bn$-nodes. 
	\end{notation}
	
	\begin{notation}\label{notation:sequents=edges}	
		Given  an edge $e$ in the module $\M$, we write  $e: F$ meaning that $e$ is labelled by the formula $ F$. As  standard, we  denote  a sequence $e_1: F_1,\dots, e_n: F_n$ of labelled edges  with the sequent $ F_1,\dots,  F_n$.
		In particular,   we   write $\R:\Delta$   for a proof-net of conclusions  $\Delta$.
		We write $\M:\Gamma\vdash\Delta$ meaning that the module $\M$ has premises $\Gamma$, and conclusions $ \Delta $.
		
			By convention, the orientation  is represented  in the  figures by drawing the edges so that they are  directed downwards.  
			
		We write  $\cut(\X,\Xn)$ 
		to denote a $\cut$-node of premises $\X,\Xn$. Similarly for  $\ax(\X,\Xn)$.
	\end{notation}
	
	Notice that any module can  be completed in a \pn. 
	\begin{definition}\label{def:completion}
		The \textbf{completion} $\overline \M$ of a module $\M: \Gamma \vdash \Delta$  of premises $\Gamma$ and conclusions $\Delta$ is obtained by completing every pending premise $e:F$ with a node   $\ax(F,F^\bot)$.  $\overline \M$  is a \emph{proof-net} of conclusions $\vdash \Gamma^\bot,\Delta$.
	\end{definition}

	\subsubsection{Normalization and  normal forms}
	

		\Cref{fig:rewriting} sketches the standard \MLL{} $\cut$-rewriting steps plus the structural   rules for contraction and weakening. Reduction preserves  switching acyclicity and formula labeling. 

		\begin{definition}[Reduction, expansion and normal forms] \label{def:rewriting} The normalization rules on proof-nets are given in \Cref{fig:rewriting}. 
			We speak of  atomic rules  ($\arule$-rules) for  $\arule\in\{ \axrule,    \neut, \assoc, \cid\}$.
			%
			%
			Each  rewriting rule $r$ in \Cref{fig:pn} defines a binary relation $\Red[r]$ on proof-nets, called a \emph{reduction} step  and  written $\R\Red[r] \R'$ ($\R$ $r$-reduces to $\R'$). The inverse step is called an \emph{expansion}. So if $\R\Red[r] \R'$, then  $\R$ is an   $r$-expansion of $\R'$.
			We write $\Red[]$ for $(\Red[\arule] \cup \Red[\otimes/\parr] )$.
			We say that $\R$ is $r$-normal if there is no $\R'$ such that $\R\Red[r] \R'$. 
			We write simply \emph{normal} for   $\Red[]$-normal.
		\end{definition}

		\begin{example}The \pn  in \Cref{fig:ABCDE2} (where   $\R$ is as in \Cref{fig:ABCDE1}) reduces to the \pn of \Cref{fig:ABCDE_D}, which is normal.
		\end{example}
		
		\begin{remark}[Normal forms]\label{rem:normal}
Since we have boxes, a {normal} \pn  can still 
contain cuts, weakening, and  contractions.
		\end{remark}
		
		Reduction $\Red[]$  is  terminating and confluent. 
		
%

	\SLV{}{	
		\begin{remark}\label{rk:normal_atomic}
			If 	  $\ProofN$ is a   normal \pn with atomic conclusions, then $\ProofN$ is  atomic.
		\end{remark}
		\begin{lemma}[normal proof-nets ]\label{lem:normal_net}
			Let  $\R:\Delta$ be a  \emph{normal} \pn. Then 
				(1.)  the premises of each $\cut$-node are atomic, 
				(2.) the  positive premise  of each cut-node is  conclusion of an $\Ax$-node. 
		\end{lemma}}
		\condinc{}{	\begin{lemma}[normal proof-nets ]\label{lem:normal_net}
				A  proof-net $\R$  which is \emph{normal} (\ie $m$-normal)  satisfies the following properties: 
				\begin{enumerate}
					\item  the premises of each $\cut$-node are atomic
					\item the  positive premise  of each cut-node is  conclusion of an $\Ax$-node
					\item the negative premise of any  cut-node is conclusion of either an $\Ax$-node or a $@$-node or a $\w$-node
					\item the conclusion of any $\weak$-node is either conclusion of $\R$, or premise of a cut.
				\end{enumerate}
				
		\end{lemma}}

		\condinc{}{\begin{lemma}
				If $\R$ is $\Red[m,s]$-normal,  it contains no $\weak$-node.
			\end{lemma}
		}

\subsubsection{Decompositions.}
We will extensively use the following decomposition of a proof-net in its set of  boxes {$ \Ax(\R) $} and the remaining  (purely \MLL) module. 
\begin{notation}[\MLL decomposition of proof-nets]\label{not:module} 
	
	%
\textbf{	(1.)}
	Let  $\R$ be  a  proof-net, and    $\Nets=\{\N_1:\Gamma_1, \dots, \N_h:\Gamma_h\}$  a \emph{set} of  distinct sub-nets   which include all the $\bn$-nodes of $\R$. Then $\Nets$
	induces a \textbf{decomposition} of $\R$, denoted by  $\R=\pair \N\M$, where $\M$ is the \MLL module which is obtained by removing from $\R$ the sub-nets in $\Nets$ (except   the edges in  $\Gamma_1, \dots, \Gamma_h$).
	Please observe that   the edges    $\Gamma= \Gamma_1, \dots, \Gamma_h$  are the pending premises of $\M$. With a slight abuse of language, we sometime refer to $\Gamma$ as the conclusions of $\Nets$.\\
%
\textbf{	(2.)}	Every    proof-net $\R$  can be decomposed  as $\pair {\Ax(\R)} {\nax \R}$. 
	
	%
\end{notation}	
	 In Linear Logic, it is standard to decompose a \emph{cut-free}  \MLL proof-net of conclusion $F$   into the sub-module constituting the syntax  tree of $F$,  and a set of axioms on the atomic formulas. A similar result holds  for any $\MLLAx$ \pn which is \emph{normal}, taking into account \Cref{rem:normal}.
	\begin{property}\label{fact:canonical_dec}If  $\R$ is a \emph{normal}  proof-net of conclusion  $F$, and $\Gamma$ the sequence of its atomic subformula, 
		 then 
		$\R$	 can be decomposed in an \emph{atomic proof-net} $\R'$ of conclusions $\Gamma$, and the module  consisting in the   syntax  tree of $F$.  A similar result hold if $\R$ has conclusions $F_1, \dots, F_k$.
	\end{property}
\begin{example} In \Cref{fig:ABCDE1}, the syntax tree of the (only) conclusion $C^+\otimes D^+$  is   red-colored.	
\end{example}
	Because of this, and w.l.o.g., in the rest of the paper   we   focus  on atomic proof-nets.
	


\subsubsection{Atoms and internal atoms.}

\begin{notation}[$\Var{\R}$]	
	\SLV{}{Two atomic  edges $e_1, e_2$ in a module $\M$ are \emph{\supported by the same \atom}  if their labels have the same support. For example, $e_1:\X$ and $e_2:\Xn$.
}
	
	An atom $X$ \emph{appears} in the formula  $F$ if $\X$ or $\Xn$ are atomic sub-formulas of $F$. 		
	We denote by		$\At{F} $ (resp. $\At{\M}$)  the set of \atoms which appear in $F$  (resp., in the formulas labelling $\M$). E.g.,  $\At{\X}=\At{\Xn}=\{X\}$.

\end{notation}

\begin{definition}[Internal atoms]
	Let $\R:\Gamma$ be a proof-net. 
	An \atom $X\in \At{\R}$  is  \emph{internal} if it does not appear in the conclusions of $\R$, \ie $X \in (\At{\R} - \At{\Gamma})$.
	Notice that if $Y\in \Var \R$ is internal   $\R$ contains  at least one node  $\cut(\Y,\Yn)$.
\end{definition}

\begin{remark}[Weakening as {hiding}.] \label{rem:weak}
If a \pn $\R$ has conclusions $\Gamma$, then 
	any \emph{positive} atomic formula $\Y$ appearing in $\Gamma$  can be hidden  from the conclusions (hence made internal) by opportunely
	composing (via cut) $\R$ with a proof-net where $\Yn$ is introduced by a $\weak$-node (\Cref{fig:deweak}).
	\condinc{}{	Formulated in sequent calculus:
		\begin{center}
			$	\infer[\cut]{\vdash \Gamma'}{\infer{\vdash \Gamma', \Y}{\dots}  &\infer[\w]{\vdash \Yn}{}}$
		\end{center} 
		Routine  operations  allow us to access the 	
		atomic sub-formulas	of a conclusion $F$.
	}
	%

\end{remark}

\begin{example}In \Cref{fig:ABCDE2} we  obtain the \pn $\R_D:D^+$ of  $\Cref{fig:ABCDE_D}$  from the  proof-net  $\R:C^+\otimes D^+$ of  \Cref{fig:ABCDE1}. 
	
\end{example}

\begin{figure}
	\centering
	\fbox{	\includegraphics[page=2,width=1\linewidth]{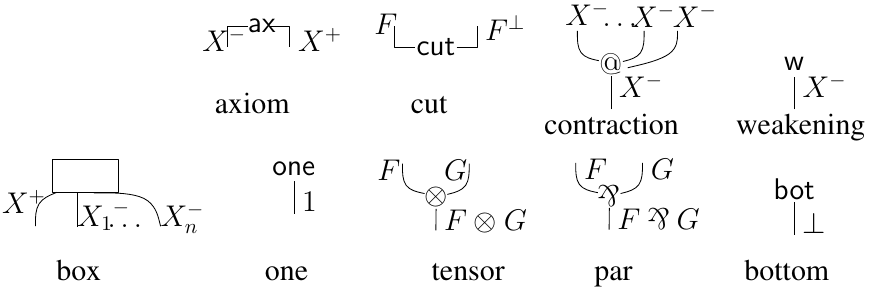}\vspace*{-8pt}
	}	\caption{Grammar of nodes} \label{fig:pn}
	\vspace{-15pt}
\end{figure}
\begin{figure}
	\centering
	\fbox{	\includegraphics[page=3,width=1\linewidth]{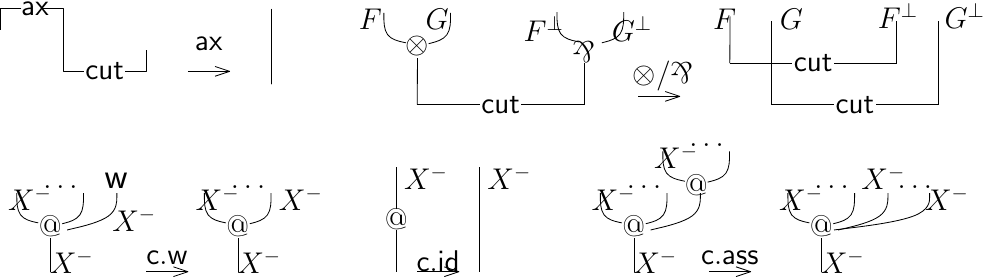} \vspace{-8pt}
	}\caption{Reduction Rules}
	\label{fig:rewriting}
\end{figure}
\vspace{-8pt}
\begin{figure}
	\centering
	\begin{minipage}[b]{4cm}
		\includegraphics[page=1,width=1\linewidth]{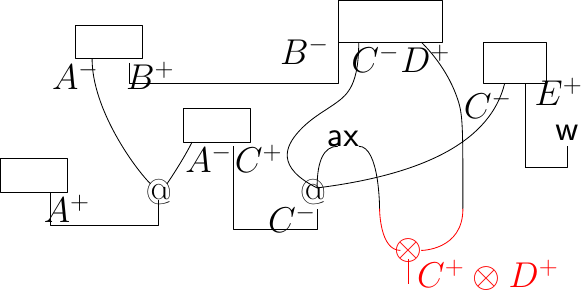}
		\caption{$\R$}
		\label{fig:ABCDE1}
	\end{minipage}
	\quad
	\begin{minipage}[b]{4cm}
		\includegraphics[page=2,width=1\linewidth]{FIGS/Fig_Example}
		\caption{Hiding $C$ in $\R$}
		\label{fig:ABCDE2}
	\end{minipage}
	\vspace{-15pt}
\end{figure}
\begin{figure}
	\centering
	\fbox{\includegraphics[page=4,width=1\linewidth]{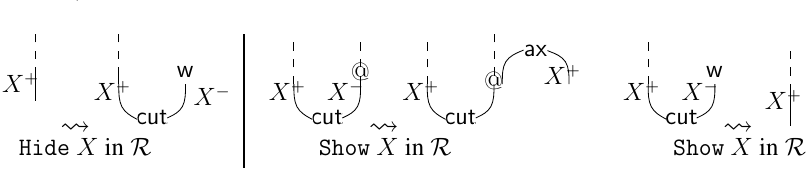}}
	\vspace*{-8pt}
	\caption{Hide  from conclusions, and Show}
	\label{fig:deweak}
\end{figure}

\section{Bayesian Proof-Nets}\label{sec:bpn}

A
\pn is composed of  a set of boxes $\Ax{(\R)}$,  and of a \emph{purely multiplicative} module $\nax{\R}$ (\Cref{not:boxes}).  Intuitively, the boxes are just containers, in which to store semantical information. We call  \emph{interface} of a \pn  its conclusions (the external interface), and the conclusions of the boxes (the internal interfaces). 
We define  Bayesian a class of proof-nets 
   whose \emph{interfaces} satisfy  the typing conditions in \Cref{def:bpn}.
It is  evident that every \BN---seen as a proof-net with the obvious translation---satisfies such  conditions. More interestingly, we can associate a \BN to any proof-net which satisfies the interfaces conditions. To formalize this, we need the notion of polarized order (\Cref{sec:polarized}). The interface  conditions have  a graph-theoretical (and proof-theoretical) meaning 
that will be the topic of   \Cref{sec:MLLcomponents}, and that will be crucial to relate  with   Bayesian inference.

	\begin{definition}[\repfree]  
		\condinc{}{If $\Gamma$ denotes a sequence of atomic edges, it is   {\emph{\repfree}} if no two labels are supported by the same \atom.		
		$\Gamma$ is \emph{positive} if all labels  are positive or if it is \emph{empty}, negative otherwise. }

		  A sequent  
		$\Gamma$ 
		is   \emph{\repfree} if no two occurrences of atomic sub-formula  are supported by the same \atom
	(so, the sequent  $ \X,\Xn$ is \emph{not} \repfree).	\end{definition}
\vspace*{-6pt}
		\begin{definition}[Bayesian proof-net]\label{def:bayesian_pn}\label{def:bpn}
		A  proof-net $\R:\Delta$ is   \emph{Bayesian}  whenever  the following \emph{interface conditions} hold. 
		\begin{itemize}
			\item[1.] 	 ${\Delta}$ and the conclusions of each box in $\R$ are   \repfree sequents.
			\item[2.] 	  If $\X$ is 
			 the positive conclusion of a  box, then no other box has conclusion $\X$,  and  $\Xn$ does not appear in $ \Delta$.
		\end{itemize}
	We write  \textbf{\bpn} for ``\emph{atomic} Bayesian \pn".
		A   \bpn  $\R:\Delta$ is \emph{positive} if $\Delta$ is \emph{empty}, or if all its atomic formulas  are  positive.
	\end{definition}
	Condition (2.) states that 	each atom $X\in \At{\R}$
	appears either positive in the conclusion of exactly one box, or  negative in the conclusion of $\R$. 
\begin{notation}
		We   denote each box in  a Bayesian \pn $\R$ by the atom which labels its unique positive conclusion (for example,   $\bbox^X$ for the box of  positive conclusion $\X$).
\end{notation}
	

 Observe that if $\R$ is a Bayesian \pn in normal form, 
	\Cref{fact:canonical_dec} applies. The atomic sub-net obtained by removing the syntactic tree of the conclusions is necessarily  a \bpn.
	 In the rest of the paper, we focus our  analysis on  \bpn's.
	 

\subsection{{Atomic proof-nets}	 and polarized order }\label{sec:polarized}
Atomic proof-nets are a simple case of \emph{polarized} proof-nets \cite{Laurent03}, where every edge is labelled by a positive or a negative  formula.
The following is an easy property of polarized nets.

\begin{lemma}[Polarized correctness \cite{phdlaurent}]\label{lem:pol_correct} For $\M$  an atomic   pre-module, we denote by  $\pol{\M}$  
	the  graph which has the same nodes and edges  as $\M$,  but where the edges are directed 
	\SLV{}{according to their polarity:} downwards if positive,  upwards if negative. 	
The following are equivalent: (1.)
  $\M$ is acyclic correct, (2.)
	 $\pol \M$   is a DAG.
\end{lemma}		
\SLV{}
{\begin{proof}We only prove $2. \Rightarrow 1.$ (the other direction is immediate). Recall that a DAG where each node has at most an outcoming edge is a tree. By inspecting the grammar of nodes, we observe that only $\parr$ and $@$ nodes may have more than one outcoming edges (outcoming for the polarized orientation). Hence each switching graph is a tree. 
\end{proof}}

\condinc{}{If moreover $\R$ is weakening-free, the following are equivalent:
	\begin{enumerate}[i.]
		\item $\R$ is \CC
		\item $\pol \R$   has a single positive conclusion (and possibly several negative conclusions).
	\end{enumerate}
}

For   atomic proof-nets, the polarized orientation  induces a partial order on the nodes, and  in particular on $\Ax(\R) $, with 
$ \bbox_1<\bbox_2 $ if there is a polarized path  from $ \bbox_1$ to $ \bbox_2$.

\begin{definition}[Polarized order] If   $\M$ is  an atomic module (hence acyclic correct), 
the polarized orientation induces a partial order on the nodes and edges of $\M$, with $x \leq y$ if in $\pol{\M}$ there is an oriented path from $x$ to $y$. We refer to this order as the \emph{polarized order}.
And edge (a node) is \emph{\textbf{initial}} (final) if it is minimal (maximal) w.r.t. the polarized order.	
\end{definition}

\condinc{}{
\begin{remark}If $\M$ is an atomic module, a positive (negative) conclusion of $\M$  is always final (initial). 
	Conversely,  a positive (negative) pending premise is always initial (final).
\end{remark}
}


The set $\Box(\R)$ with the induced partial  order is am  \emph{invariant} of the proof-net under several 
important transformations, and in particular under normalization. 
\begin{definition}[$ \bnet{\R}$]\label{def:core} For $\R$ an atomic proof-net, we denote by  $\bnet \R$ 
	the DAG over  $\Ax{(\R)}$ induced by the polarized order.
\end{definition}
\SLV{}{We refer to $ \bnet{\R} $ also as the \emph{core invariant} of $\R$. }
\begin{lemma}
	$\bnet{\R}$ is invariant under $\arule$-reductions and $\arule$-expansions.
\end{lemma}

\subsection{Bayesian proof-nets and Bayesian Structures}\label{sec:bpn_BN}

Let $\R$ be an atomic \bpn.
 Since to  each box  corresponds exactly one positive occurrence of atom, 
the  correspondence between \bpn's and  Bayesian structures is  immediate\footnote{Here, we do not say yet what is the space of values associated to the variable $X$. This will come with the interpretation. The values sets associated to the  \rvars is however irrelevant when reasoning with Bayesian Structures. }. 
\begin{property}
\begin{enumerate}
	\item 	Let   $\R:\Gamma$ be a  positive \bpn and  $\At{\R}=\{X_1, \dots,X_n\}$. Then	$\bnet{\R}$ is isomorphic to a  Bayesian Structure  $\netG$ over   $X_1, \dots,X_n$.   Notice that   ${\Gamma}= \X_1, \dots,\X_n$ when   $\R$ has no internal atoms.  
	\item Conversely, given a Bayesian Structure  $\netG$ over   $X_1, \dots,X_n$, there exists   a 
	   \bpn (of conclusions $\X_1, \dots,\X_n$) such that  $\bnet{\R}$ is isomorphic to $\netG$. 
\end{enumerate}
\end{property}
\begin{example}Consider the \pn $\R$ in  \Cref{fig:ABCDE1}. It is  clear that $\Bnet{\R}$ is exactly the DAG of the BN in \Cref{fig:BNrain}.
\end{example}
Point (1.)  is evident when $\R$ is in normal form, but  it holds also when  $\R$ is not normal, perhaps having a complicated structure, because   $\bnet \R$  is  invariant  under normalization.

Clearly, given a  Bayesian Structure  $\netG$, there are several \bpn's $\R$ such that  $\bnet{\R}$ is isomorphic to $\netG$.  The obvious canonical choice is   the one in normal form and with no internal atom.
\SLV{}{The choice about the  conclusions  will be justified in \Cref{sec:PR}. }
Notice  that $\bnet \R$ is independent from the conclusions of $\R$.


\subsection{Weakening and de-weakening.}

Proof-nets which have the same  $\bnet \R$ 
may have very different conclusions, spanning from no internal atom, to empty conclusions.
We observed (\Cref{rem:weak}) that weakening allows us to hide \emph{positive} atomic formulas. 
Such an operation is reversible. Elementary $\MLL$ manipulations produce  a dual operation which given an internal \atom $Y$, produces a conclusion labeled by $\Y$.

\begin{lemma}\label{def:de-weakening} Let  $\R$  be  a  \bpn.

\textbf{\textsf{Hide}.} For $\R:\Gamma',\Y$, the operation 
  ${\OHide}$ consists in cutting the conclusion $\Y$ with a $\w$-node  (\Cref{fig:deweak}).
		The resulting proof-net $\W\R Y$ is a \bpn conclusions $\Gamma'$ ($Y$ is  internal).
		
	\textbf{\textsf{Show}.} 
	Assume  $Y$ is   an internal \atom of $\R:\Gamma$. 	Necessarily, $\R$ contains  at least one node $\cut{(\Y,\Yn)}$.
		The \emph{de-weakening} of  $Y$ in $\R$ is the operation ${\OShow}$ described in \Cref{fig:deweak}, where  we assume that the negative edge $e:\Yn$ is conclusion of a $\cn$-node (which is always possible by $\neut$-expansion). 
		The resulting 
		proof-net $\deW \R {Y}$ is a \bpn of conclusions $\Gamma, \Y$. 
		We call \emph{full}  de-weakening of $\R$  the result  of repeatedly applying the operation to  all   internal \atoms of $\R$. 

\end{lemma}

\condinc{}{
	\begin{remark}
		Assume that the cut-node has premises $f:\Y, e:\Yn$. 
		\begin{enumerate}[i.]
			\item   $e:\Yn$ is  conclusion of a $w$-node: remove  it and the cut, so that $f:\Y$ becomes a conclusion of the proof-net.  
			\item   $e:\Yn$ is  conclusion of a $@$-node: 
			add an axiom $\ax(\Y,\Yn)$, making its  negative edge  premise of the $@$-node.  The positive edge becomes a conclusion of the proof-net. 
			\item  $e:\Yn$ is conclusion of an axiom ($\ax$ or $\Ax$):  insert one $@$-node in $e$ (note that this is an $m$-expansion rule), and procede as in (ii).
		\end{enumerate}
	\end{remark}
}
\begin{example}
For $\R_D$ as in \Cref{fig:ABCDE_D}, $\Hide {\R_D} C$ produces the black-colored  \bpn in \Cref{fig:ABCDE1} (please ignore  the part in red).
\end{example}
The following properties are easy to check.
\begin{lemma}
	Let $\R$  be a \bpn.
 \begin{itemize}
 	\item $\bnet \R $ is invariant under  $\OHide$ and ${\OShow}$.
\item  $\Show {\Hide {\R} X} X    \Red[]^* \R$.
\item $\Hide {\Show {\R} X} X   \Red[]^* \R$.
 \end{itemize}
\end{lemma}

\vspace*{-4pt}
\section{\MLL  connected components: wirings}\label{sec:MLLcomponents}


A \bpn is composed of  a set of boxes $\Box{(\R)}$,  and of a \emph{multiplicative, atomic} module $\nax{\R}$. The latter  acts as a \emph{wiring}, plugging  together and carrying around the  information  stored in the boxes. In this section we identify  
a property of \emph{\goodness} which  guarantees that the behavior of the wiring is as expected   (see \Cref{sec:PCoh}). We then prove that  the interface conditions  in \Cref{def:bpn} are equivalent to  \goodness.

Such a property    is  essential to the  interpretation of \bpns. Remarkably 
(\Cref{sec:pt})  the  same   property will turn out to be the proof-theoretical counterpart of  the characterizing  property  of \emph{cliques trees}, the data structure underlying  the message passing algorithm for Bayesian inference.

%
%
%
%
%

\begin{definition}[Atomic \MLL module] 
	We refer to an atomic module $\M$ which  only contains  nodes of sort $\ax,@,\weak,\cut$ (and  no  $\bn$-node) as an atomic \MLL module.
\end{definition}

\condinc{}{
\begin{definition}[\good (wirings)]\label{lem:support}  
 An atomic  \MLL module  $\M$  is said   \emph{\good} if 
	any   pair of edges which are supported by the same \atom are connected. 
	
	Stated otherwise, when  partitioning $\M$ in maximal connected components, (the edges of) \emph{distinct connected components are supported by  distinct \atoms}, that is:
\begin{center}
		 \emph{one component = one \atom}.
\end{center}

We call \emph{\wiringM} (resp. \emph{\wiring}) a module (resp. a proof-net) which is atomic and \good.
\end{definition}	
}
\begin{property}[Initial edges]\label{fact:initial}
The initial edges of an atomic \MLL module (\ie the minimal edges w.r.t. the polarized order) are its negative conclusions, and its positive pending premisses.
\end{property}
The proofs of the following  statements are in \Cref{app:MLL}.
\begin{lemma}\label{lem:connected}Let $\M$ be an atomic \MLL module  which is a connected graph. Then:
		(1.) 	All edges are \supported by the same \atom. 
		(2.) $\pol \M$---\ie $\M$ with the polarized orientation---is a directed tree, with a \emph{unique} initial edge.
\end{lemma}
	\begin{proposition}\label{lem:M_connected}  	
		Let $\M: \Gamma \vdash \Delta$ be an atomic \MLL module (of conclusion $\Delta$,  premises $\Gamma$). Then  $\At{\M} = \At{\Gamma} \cup \At{\Delta}$, and 
			the following are \emph{equivalent} for each  $X\in \At{\M}$:
\textbf{(1.) }	 $\M$ has exactly one initial edge supported by $X$. 
		\textbf{	(2.)} $\M$ has exactly  one maximal connected component supported by $X$.

	\end{proposition}
\condinc{}{
\begin{corollary}Let $\D:\Delta$ be an    atomic \MLL proof-net. $\D$ is a connected graph if and only if $\Delta$ contains exactely one negative conclusion (and any number of positive conclusions.)
\end{corollary}
}
\begin{theorem}[Characterization]\label{thm:char}Let $\M: \Gamma \vdash \Delta$ be an atomic \MLL module. The following are equivalent:
	\begin{itemize}
		\item[1.]	Any   two  edges supported by the same \atom are connected. 
		\item[2.]  When  partitioning $\M$ in maximal connected components, (the edges of) \emph{distinct connected components are supported by  distinct \atoms}.
		\item[3.] Distinct initial edges of $\M$ are supported by distinct \atoms.
		\item[4.]  The  \emph{positive} atomic formulas in $\Gamma,\Delta^\bot$ are pairwise distinct.
	\end{itemize}
	
\end{theorem}
\begin{proof}$ 1. \Leftrightarrow 2. $ and $  3. \Leftrightarrow 4.$ are immediate. $ 1. \Leftrightarrow 3. $ follows from    \Cref{lem:M_connected} .
\end{proof}
So the above are equivalent formulations of the property  
\begin{center}
	\emph{one connected component = one \atom}.
\end{center}
We call such a property \emph{\goodness}, and we call \emph{wiring} an atomic \MLL module  which satisfies such a property.

\condinc{}{
\begin{cor}\label{cor:connected}  	
		Let $\R$ be an atomic \MLLAx proof-net  
		which decomposes as $\pair \Nets \M$,
		for $\Nets$ is a proof-net of conclusions $\Gamma$, and  $\M$
		  an \MLL module. 
%
		The following are equivalent:
			\begin{enumerate}[i.]
				
%
			 \item Distinct positive conclusions of $\Nets$ (resp., distinct negative conclusions of $ \R $) are labelled by distinct atoms. Moreover, if   $\X \in \Gamma_1, \dots, \Gamma_n$ then   $\Xn \not\in  \Delta$.  
%

				\item $\M$ is \good.
			\end{enumerate}
	\end{cor}
	\begin{proof}
		\condinc{}{Let  $\netC$ be a maximal connected component in $\M$
		(by \Cref{lem:connected} all edges are supported by the same \atom, and 
		$\pol{\netC}$ is an oriented tree). we observe that the   unique  initial edge of  $\pol{\netC}$  is either the positive conclusion of some $\N\in \Nets$,  or a negative conclusion of $\R$.}
		By construction and \Cref{fact:initial}, 
		the initial edges of $\M$ (with the polarized orientation) are its positive pending premises  (the positive edges in  $\Gamma_1, \dots, \Gamma_h$)
		and its  negative pending conclusions (the negative edges in $ \Delta $).
		The claim therefore follows from \Cref{lem:M_connected}.
	\end{proof}
}

%

\begin{definition}[\good]\label{lem:support} \label{def:wiring}
\begin{itemize}
	\item 	An atomic  \MLL module  $\M$  is said   \emph{\good} if 
	any   pair of edges which are supported by the same \atom are connected. 
	
%
\item 	We call \emph{\wiringM} (resp. \emph{\wiring}) a module (resp. a proof-net) which is atomic and \good.
	
\end{itemize}
\end{definition}	
  \Cref{thm:char} implies the following, which motivates \Cref{def:bayesian_pn}.
\begin{cor}\label{lem:B_connected}  	
	Let $\R:\Delta$ be an atomic \MLLAx proof-net such that its conclusions $\Delta$ and  the conclusions of each $\Ax$-node are \repfree sequents. The following are equivalent:
	\begin{itemize}
		\item 	 $\R$ is a \Bpn 
		\item the module $\nax{\R}$ 
		is \good.
	\end{itemize}
\end{cor}

We conclude  with a slogan which   underlies  the rest of this paper:\quad	\textbf{\textit{Bayesian proof-nets = boxes +  wiring}}

\condinc{}{
\todocf{Proof-theoretically, an  \MLL component corresponds to the notion of occurrence of a formula, as used in sequent calculus cut-elimination and in Geometry of Interaction literature. Intuitively, atoms can be renamed in order to satisfy the interaces conditions, recalling that 
	  propositional variables stand for unknown propositional formulas (which need not be different formulas, but will be different occurrences). The formal  notion of renaming is  developed in the sibling paper \cite{licsA}.}
}

\section{BN's and Proof-Nets: Representation}\label{sec:PR}
So far, we only worked with the qualitative structure.  What is still missing  is the quantitative information to be associated to  the boxes. Such a content is given  by the semantics of Linear Logic. We first present PCoh, together with  a "recipe" to compute the interpretation of a \bpn. 
 We then complete the picture started  in  \Cref{sec:bpn_BN}. In  \Cref{sec:cost}, we   analyze the cost of the PCoh interpretation. Such an analysis will motivate  the factorization of   proof-nets in the following sections.

\vspace*{-4pt}
\subsection{On the Quantitative Semantics: PCoh}\label{sec:PCoh}
Probabilistic coherence space semantics (PCoh) \cite{Girard2003, danosehrhard} is 
a  quantitative semantics particularly suitable for modelling probabilistic programs. 
We refer  the reader to standard literature  \cite{panorama, EhahrdPT18fa, EhrhardT19} for  details---  we only mention what we need.
\vspace*{-4pt}
\begin{definition*}
	A \textdef{probabilistic coherence space} (PCS) is a pair
	$\Pcs{X}=(\Web{\Pcs{X}}, \Clique{\Pcs{X}})$ where $\Web{\Pcs{X}}$ is a
	finite set 
	and
	$\Clique{\Pcs{X}}$ is a subset of $\RP^{\Web {\Pcs{X}}}$ satisfying  conditions which allow for the involutive negation of linear logic.

	The \textdef{objects} of the category $\PCoh$ are PCS's and the set 
	of \textdef{morphisms from $\Pcs{X}$ to $\Pcs{Y}$} is the set of 
	matrices $\MatA\in \RP^{{\Web{\Pcs X}}\times\Web{\Pcs{Y}}}$ such that 
	$ 		\forall \VecA\in \Clique{\Pcs{X}}, \App  \MatA \VecA \in \Clique{\Pcs{Y}}$. 
	\textdef{Composition} is matrix multiplication,  \textdef{identity} is the identity matrix.
\end{definition*}
The elements in $\Web{\Pcs X}$ represent possible values, while the vectors in $\Clique{\Pcs X}$ generalize the notion of discrete probabilistic distribution.
Formulas and sequents  of multiplicative Linear Logic are  interpreted as   PCS's.
%
%
A  proof  is interpreted as a morphism, or
---equivalently---as a vector: given an interpretation $\Ev$ for the atoms and for the generalized axioms (the boxes),  the  interpretation  of a \pn $\R$ is  defined (by structural induction on the sequent calculus derivation) as 
a  vector $\Sem[\Ev]\ProofN\in\Clique{\Sem[\Ev]{\Delta}}$, where $\Sem[\Ev]{\Delta}$ is the interpretation of   $\Delta$.


\paragraph{ PCoh Interpretation of  a \bpn: a recipe.} Given a \bpn, 
the natural PCoh interpretation of boxes are CPT's (see \Cref{def:CPT}). 
The PCoh interpretation of multiplicative proofs can  be formulated in terms of factors---the development of this formulation in the general  setting is the object of a sibling paper.
%
The formulation in the case of   \bpns is however especially simple, and we state \Cref{thm:PCoh}\footnote{This result is proved in a sibling paper, in a general setting of which \bpns are a special case. The formal  argument relies on an analysis of the connected components. Thanks to  the good labelling property, in the case of  \bpns, \Cref{thm:PCoh}
 can  also be proved directly, in a straightforward way,   by induction on a smartly chosen sequent derivation,  where contraction is greedy.
} directly in the setting of \bpn's. 
For the purpose of this paper, the  reader who is not familiar with PCoh, can simply take \Cref{thm:PCoh} as an operational recipe to compute $\phisem \R$.
%
%
The PCoh interpretation of a \bpn  $\R$ is 
fully determined  as soon as we provide a  \groundinterpretation $\iphi$, as follows.
\begin{definition}[\groundinterpretation]
Given a \bpn $\R$, a   \groundinterpretation $\iphi$ is a mapping that associates 
\begin{itemize}
	\item to  each \atom $X\in \At{\R}$ a  values set  $\iphi (X)= \Val X$, and 
	\item to each box $\bbox^X$   of conclusions $ \X, \Yn_1, \dots, \Yn_k$ a \CPT  $\iphi(\bbox)= \phi^X(X| Y_1, \dots, Y_k)$
	which is a \CPT for $X$ (\Cref{def:CPT}).
\end{itemize}
\end{definition}
\begin{theorem}\label{thm:PCoh}
	Let $\R:\Delta$ be a \bpn and $\iphi$ a \groundinterpretation. 
	\begin{center}
		$ 		\phisem{\R} = \FProj{\BigFProd_{X\in \At{\R}}\iphi(\bbox^X)}{\At{\Delta}} $
	\end{center}
\end{theorem}
Notice what the interpretation is doing here: first, it computes the product of all the factors $\iphi(\bbox^X)$, then  projects on the \atoms 
which support the conclusions.  
Intuitively, the  multiplicative structure  performs the product of the factors, by connecting together all the edges labelled by the same atom (as represented in \Cref{fig:ABCDE3}).

We will use extensively  the fact that  the PCoh interpretation is invariant  under normalization.
\begin{prop*}Let   $\R:\Delta$ be a \bpn, and $\iphi$ a \groundinterpretation.\\ If $\R:\Delta\Red[]^*\R':\Delta$ then
$\phisem \R =\phisem{\R'}$.
\end{prop*}

We  have everything in place to   explicit   the correspondence between   \bpns (equipped with a \groundinterpretation) and  BN's.

\subsection{Bayesian Networks and Proof-nets}\label{sec:representation}

\condinc{}{
In  this section, let 
  $\R$ be an    arbitrary but fixed  $\bpn$ of \emph{positive} conclusions $\Gamma$, for which we   fix a \groundinterpretation $\iphi$.}

Given  $\R$  a \bpn and  $\iphi$ a \groundinterpretation, 
we can see each atom $X\in \At{\R}$ as a variable (a \rvar in the sense discussed in \Cref{sec:basics}) assuming values in $\Val X$. 
 From now on we silently identify the set of atoms $\At{\R}=\{X_1, \dots, X_n\}$ with  the set of variables $\{X_1, \dots, X_n\}$ having respective  value sets $\iphi(X_1)=\Val {X_1}, \dots, \iphi(X_n)=\Val {X_n}$.  \\

Recall that  $\bnet \R$ denotes the DAG over  $\Ax{(\R)}$ induced by the polarized order (\Cref{def:core}).
\begin{property}Let $\R$ be a positive  \bpn, $\iphi$ a \groundinterpretation, and   $\bX=\{X_1, \dots, X_n\}$ be (the set of \rvars corresponding to) $\At{\R}$.
	\begin{itemize}
		\item  The pair $(\bnet{\R}, \iphi)$ is isomorphic to a  \BN over  $\bX$, which we denote $\Bnet{\R, \iphi}$.
		\item Conversely, a \BN $\netB=(\netG, \aphi)$ induces a \bpn in normal form and a \groundinterpretation.
	\end{itemize}
\end{property}
As  for  $\bnet \R$, the following holds. 
\begin{prop}[Invariance] For  $\R$  a \bpn, and $\iphi$ a \groundinterpretation,  $\bnet{\R,\iphi}$ is invariant under all the following transformations of $\R$: $\arule$-reduction and $\arule$-expansion,  $ \OHide $ and $ \OShow $. 
\end{prop}

\subsubsection{Probabilistic view}\label{sec:PView}
The  BN 
associated to a positive \bpn $\R$ only depends on  $\bnet \R$, and 
the \groundinterpretation.  It  is independent from the  actual conclusions of $\R$, which  instead  specify a marginal.
If $\R$ has no internal \atoms, then its semantics $\phisem \R$ is exactly  the joint probability represented by  $\Bnet{\R, \iphi}$ (see \Cref{def:BN}).
\begin{prop}\label{prop:bpn_vs_bn}Let $\R:\Gamma$ be a  \emph{positive} \bpn,  $ \iphi $     a \groundinterpretation with $\iphi(\bbox^X)=\phi^X$,  $\bX=\At{\R}$.     
 Let   $\Pr_\netB(\bX)$ be the probability distribution over $\bX$ specified by the BN $\netB=\Bnet{\R, \iphi}$.
\begin{enumerate}
		\item \emph{Joint probability.} Assume    $\R$ has no internal atoms. \SLV{}{\ie $\At{\Gamma}=\At{\R}$.} Then 
	\begin{center}
		$ 	 \phisem \R =   \BigFProd_{X\in \bX} \phi^X \defeq  \Pr_{\netB}(\bX)   $\\
	\end{center}
	\item \emph{Marginals.}  Let $ \At{\Gamma} \subseteq \At{\R}$. Then
\begin{center}
$ 		\phisem{\R} =  \project {\BigFProd_{X\in \bX} \phi^X } {\bY}=  \Pr_{\netB}(\bY)   $
\end{center}
\ie,  $\phisem \R$ is   the marginal probability over $\bY=\At{\Gamma}$.
\end{enumerate}
\end{prop}
\SLV{}{Let $\R$ be a  \Bpn, and $\iphi$     a \groundinterpretation. Let $\Box(\R) = \{\bbox^{\X_1}, \dots,\bbox{\X_n}\}$, and $\bX =\{X_1, \dots, X_n\}$.}
\noindent
 More generally, we define (omitting $\iphi$ when irrelevant)
\begin{center}
		$ 	\Pr_{\R,\iphi} \defeq \BigFProd_{X\in \bX} \iphi(\bbox^X) = \Pr_{\Bnet{\R, \iphi}}$. 
\end{center}
By \Cref{prop:bpn_vs_bn} (with the same assumptions of conditional independence as for Bayesian Networks) $ \Pr_{\R,\iphi} $ specifies a joint probability distribution over $\bX$.  

\begin{remark}[Weakening.]
	The role of weakening for marginalization has since long been pointed out in the literature \cite{JacobsZ, Stein2021CompositionalSF}. De-weakening has a dual role.
\end{remark}	

\condinc{}{\begin{remark}[Role of weakening and de-weakening.]
The role of weakening for marginalization has already been pointed out in the literature \cite{JacobsZ, Stein2021CompositionalSF}. De-weakening has a dual role. 
 For any positive \bpn $\R$, its  full de-weakening  yields  $\Pr_{\R,\iphi}  $.
 We will use $\OShow$ also more subtly,  as a tool to simplify computations.  
 
 \SLV{}{$\OHide$ and $\OShow$ allow us to hide some  \atoms ---or  conversely to make internal atoms visible---by modifying the conclusions of  $\R$, without altering   $\bnet{\R}$. }
\end{remark}}
In the rest of the paper, we  often  write $\phi^X$ for $\iphi{(\bbox^X)}$.
\subsection{On the  cost of  interpretation by \Cref{thm:PCoh}}\label{sec:cost}
Let us now consider the cost of the interpretation  recipe given by \Cref{thm:PCoh}. Assume  the cardinality $\size{\At{\R}}$ of $\At{\R}$ is $n$.  Recalling \Cref{rem:factors_cost}, to compute  $\phisem{\R}$ 
according to \Cref{thm:PCoh} demands to compute and store $\BigO{\Exp{n}}$ values. 
However, \emph{we can do better} if we are able to compute smaller products (similarly to what  inference algorithms do on Bayesian Networks). We do so in the next sections.

\section{Factorized form of a proof-net}\label{sec:factorized_form}
\begin{figure}
	\centering
	\includegraphics[page=3,width=0.7\linewidth]{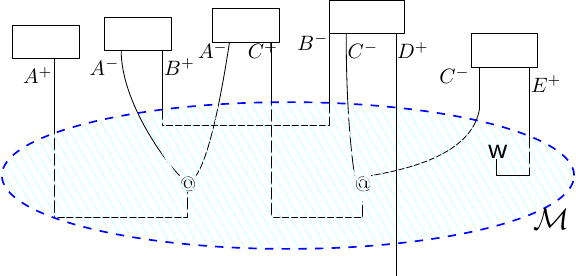}\vspace*{-5mm}
	\caption{$\R_D$}
	\label{fig:ABCDE3}
	\quad
	\includegraphics[page=4,width=0.7\linewidth]{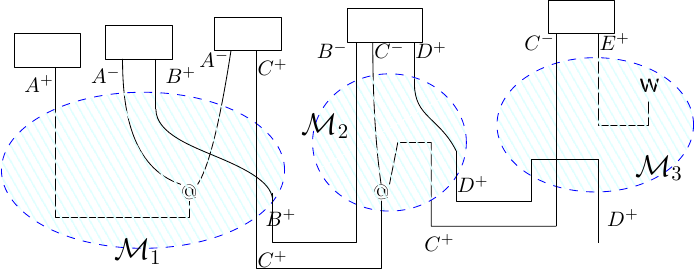}
	\caption{$\R_D'$. \quad Notice $\sem {\R_D}= \sem{\R_D'}$}
	\label{fig:ABCDE4}
\end{figure}

In this section, we factorize a proof-net in the composition  of smaller nets, whose interpretation has a smaller cost. 
In proof-theory, the natural way to factorize  a proof in smaller components is to factorize  it in sub-proofs which are then composed together (via cuts). 
\begin{example}[Roadmap]\label{ex:roadmap}	Take the \bpn in \Cref{fig:ABCDE3}. Assuming all variables are binary, the cost of interpreting $\R_D$ via \Cref{thm:PCoh} is $2^5$.
	The idea  which we  develop, is to rewrite  $\R_D$ 
into a  \bpn  like $\R_D'$ in	\Cref{fig:ABCDE4}. Notice that  $\R_D'\to^* \R_D$, and so $\phisem {\R'_D} =\phisem {\R_D}$. 
\end{example}

In \Cref{sec:MP} we will  show that there is a tight correspondence between the decomposition of  a proof-net into a cut-net (introduced below), 
and the  well-known decomposition of Bayesian Networks into clique trees.
\condinc{}{ The correspondence turns out to be so tight, that the cost of computing the interpretation of a cut-net is the same as the cost of inference on  the corresponding clique tree.}

%
%
%

\subsubsection{Cut-nets}

In the literature of linear logic, the decomposition of a proof-net in sub-nets which are connected by cuts is called a cut-net \cite{locus}.
The notion of correction graph is  similar to that of switching path. 
	\begin{definition}[Cut-nets]\label{def:cut-structure}
	We write  $\R= \CutS{\R_1, \dots, \R_n}$ to denote  proof-net $\R$ for which is given a \emph{partition}    into sub-nets $\R_1, \dots, \R_n$ and $\cut$-nodes, where  the two premises  of each such $\cut$ are     conclusions of two distinct sub-nets. 
		 The \textbf{correction graph} of $\CutS{\R_1, \dots, \R_n}$ is  the graph which has nodes 
		$\{\R_1,...,\R_n\}$ and  an edge $(\R_i,\R_j)$ exactly when there is a cut between $\R_i$ and $\R_j$.	
		  $\CutS{\R_1, \dots, \R_n}$ is   a \textbf{\cutnet}  if the correction graph is a \emph{tree}. 
		{We denote this graph $\tree_{\R_1,...,\R_n}$, or simply $ \tree_\R $ if the partition is clear from the context.}
		
		\condinc{}
		{	\item A {cut-net} is \emph{proper} where   between any pair $\R_i, \R_j$ there is at most one cut. }

%
%
\end{definition}

\begin{example}[cut-net]\label{ex:cutnet}The \bpn in \Cref{fig:ABCDE4} can be partitioned as a \cutnet with 3 components (and 4 cuts). 
	If we $\ax$-expand the  conclusions of the boxes, we can  partition the resulting \bpn  as a \cutnet with  8 components (and 14 cuts).
\end{example}
\condinc{}{
	\begin{itemize}
		\item if $\R$ is AC,  \emph{each sub-structure   $\R_i$ is AC}, because any cycle in the   switching graph of $\R_i$ implies a cycle in the switching graph of $\R$.
		\item if $\R$ is CC, \emph{each $\R_i$ is CC}, because if $\R_i$ has a switching graph which is not connected, so does $\R$.
		\item A \cutnet with  $n$ components and $n-1$ traversing cuts is a \CC cut-net.
\end{itemize}}

\subsubsection{Factorized form of a proof-net}
	Recall that 	 a 	\emph{\wiring} $\D$	is an atomic \MLL proof-net  where 	any   pair of edges  supported by the same \atom are connected.  When in normal form, 
	each such  component  has an extremely simple form.

\SLV{}{
	\begin{lemma}[normal connected components]	An atomic \MLL proof-net  $\D[X]$ which is connected 
		is  either a single  axiom, or a single    $w$-node with its conclusion, or a single  $@$-node together with its  conclusion and its premises, where  each premise completed by an axiom.

	\end{lemma}
	\begin{proof} By inspecting the rules.
		Notice that in a \emph{normal} atomic \MLL proof-net, the positive conclusion of an $\ax$ node is necessarily conclusion of the net, because it cannot be premise of a cut (which would create a redex).
	\end{proof}
}

\SLV{}{
\begin{remark}
   $\D$  is a \wiring exactly when it is a justaposition of    connected   proof-nets, each supported by a distinct \atom. We then can write $\D=\biguplus_{X\in \At{\D}} \D[X]$. 
   
   Observe that  each connected net $\D[X]$ has exactly one negative conclusion (this is obvious when $\D[X]$ is normal).
\end{remark}
}

\begin{definition}\label{def:wiring_cut}		We write $ \D(\N_1, \dots, \N_h):\Delta$  to denote a  cut-net     $\R=\CutS{\D, \N_1, \dots, \N_h}:\Delta$    where  
	\begin{itemize}
		\item $\D:\Delta'$ is a \wiring;
		\item each conclusion of each proof-net  $\N_i:\Gamma_i$ ($1\leq i\leq h$) is connected by a cut to  a  conclusion of   $\D$.
	\end{itemize}
	We call $\Gamma_1,...,\Gamma_h$ the \emph{inputs} of $\D$ and $\Delta$ the \emph{output} of $\D$.  Observe that  $\Delta' =  \Gamma_1^\bot,\dots,\Gamma_h^\bot,\Delta$.
\end{definition}

\SLV{}{
	\begin{remark}Notice that , if $\R =\D(\N_1, \dots, \N_h):\Delta$  is a \bpn, then $\Delta$ is \repfree. 
	\end{remark}
}

%


\condinc{}{
\begin{definition}	
		 Assume $\D:\Delta'$ is a \wiring, and $\N_1:\Delta_1, \dots, \N_h:\Delta_h$ are \bpns. 
		We write $ \D(\N_1, \dots, \N_h):\Delta$  to denote a  cut-net    $\CutS{\D, \N_1, \dots, \N_h}:\Delta$   where   each conclusion of each $\N_i$ is connected by a cut to  a  conclusion of   $\D$, and whose conclusions $\Delta$ are a \repfree sequent.
		
		We call $\Delta_1,...,\Delta_n$ the \emph{inputs} of $\D$ and $\Delta$ the \emph{output} of $\D$ (observe that  $\Delta' =  \Delta_1^\bot,\dots,\Delta_n^\bot,\Delta$).
\end{definition}
}

%

\begin{definition}[Factorized form of a proof-net]\label{def:factorized}\hfill
	\begin{itemize}
		\item A  \bpn which consists of a single box 
		is  \emph{\factorized}.
		\item A \bpn  $ \R= \D(\N_1,\dots,\N_h):\Delta $	is  \emph{\factorized} if 
		$\N_1:\Gamma_1, \dots, \N_h:\Gamma_h$ are  \factorized \bpn's. 
	\end{itemize}
We call $\D$  \emph{the root} of the factorized net.  
Notice that  the output $ \Delta $ of $\D$ is \repfree (by the  assumption that $\R$ is a \bpn).  Similarly  for 
each of its inputs  $\Gamma_1, \dots, \Gamma_h$.
\end{definition}
A factorized \bpn $\R$ is a \cutnet whose components are the boxes, and a set of \wirings, that we denote 
 $ \Wirings{\R} $.
\begin{center}
	$
\Wirings {\R}=
\left\{
\begin{array}{l}
	\emptyset   \quad\quad \mbox{if $\R$ is a box}\\
	\{\D\}\cup \bigcup_{1\leq i\leq h}\Wirings {\N_i}\\
	\quad \quad \mbox{if $\R=\D(\N_1, \dots, \N_h)$.}
\end{array}
\right.
$
\end{center}


	Every  \bpn has a trivial  factorized form. 
	
\begin{lemma}[Trivial factorized form]\label{lem:trivial}
		Every \bpn   $\R$ can be transformed by $\ax$-expansion into $\R'=\D(\Ax(\R))$, where $\D$ is the completion of ${{\nax \R}}$. 
\end{lemma}	
\SLV{}{
\begin{proof}
	Every   proof-net $\R$  can be decomposed as $\pair {\Ax(\R)} {\nax \R}$, and transofrmed by  $\ax$-expansion, 
	  into $\R'=\D(\Ax(\R))$, where $\D$ is the completion of   $\nax \R$.

$\D$ is  \good  because 	 $ \nax \R  $ is (by \Cref{lem:B_connected} ). 
\end{proof}

\begin{remark}\label{rem:trivial}
Notices that, if $\R$ has empty conclusion and is normal, to obtain a cut-net  $\R'=\D(\Ax(\R))$ it suffices to $\ax$-expand the negative  conclusions of the boxes, because every positive conclusion is already the premise of a $\cut$.
\end{remark}
}

\begin{example}[Factorized form]\label{ex:fact}
	 Consider the  \bpn in 
	\Cref{fig:ABCDE4}. If we complete by $\ax$-expansion the three shaded modules (let $\D_i =\overline \M_i$), we immediately have a  factorized form. By choosing $\D_3$ as root, we can write this net as
\begin{center}
		$\D_3\big(\bbox^E, \D_2(\bbox^D,\D_1(\bbox^A, \bbox^B,\bbox^C))\big)$.
\end{center}
	
	The \bpn $\R_D$ in \Cref{fig:ABCDE3}  (once we  $\ax$-expanded the conclusions of the boxes) is  a (trivial) factorized form, which has a single large wiring, corresponding to   $\overline{\nax \R_D}$ (see \Cref{def:completion}).
\end{example}

\subsubsection{Changing the root of a \factorized net}
We state here a  fact that will be crucial when we  will want  to infer each single marginal in $\Pr_\R$.
A \factorized \bpn is a \emph{rooted}  cut-net. 
Recalling that the correction graph $\tree_\R$ is a tree, one  can easily realize that 
any $\D\in \Wirings \R$  in can be chosen as  root of the same cut-net, leading to a different tree-like  description of $\R$.
\begin{lemma}[Root change]\label{lem:root_change}
	Let $\R$ be a \factorized \bpn of empty conclusion.  Any choice of $\D_i\in \Wirings \R$ as root  induces a \factorized \bpn  of root $\D_i$.
\end{lemma}

\begin{example}[Root]\label{ex:root} Consider again the \bpn in \Cref{fig:ABCDE4}, with the same $\ax$-expansion as in \Cref{ex:fact}. 
	Let $\N$be the  net obtained by hiding $D$. We can write this net in three different way, and for example choose as root $\D_2$.
	Notice that $\deW{\N}{C}$  has conclusions  $C^+$ (which is conclusion of $\D_2$).
\end{example}

\subsubsection{Width}
The crucial parameter to determine the cost of interpretation, and so the quality of a factorization, is the
number of atoms in each components, which we call width (in analogy with  similar notions in BN's.).
\begin{definition}[Width of a \factorized net]\label{def:width_R} 
	We define the size 
	 of a \wiring $\D$ as the cardinality of $\At{\D}$.
	We define  $\Width{\R}$ of a \bpn $\R$ in \factorized form as the largest size of any of its \wirings, minus one.
\end{definition}
\begin{example}The factorization  in \Cref{fig:ABCDE3} has width $4$,  the factorization in \Cref{fig:ABCDE4} has width $2$. 
	Interpretation is the same.
\end{example}
\begin{lemma}
	 	$\Width{\R} =  \Width{\deW \R Y}$
\end{lemma}
\begin{proof}
The 
operation 	$\OShow$  does not change the size of any \wiring.
\end{proof}

%
%
%

%% file: 03_PN_Inference.tex

\section{Inference by Interpretation, efficiently}
Let consider again the cost of the PCoh interpretation (\Cref{sec:cost}) taking advantage of the factorization of \bpn's. We  inductively  interpret  a (factorized) \bpn  following its tree-shape.

 \Cref{thm:PCoh} easily generalize to \Cref{lem:PCoh}, using the properties of sum and products, and  \goodness.
\begin{lemma}\label{lem:PCoh} Let  $\R:\Delta$  be  a \bpn which can be decomposed as $\pair \Nets \M$\footnote{Recall  that $\M$  is the  \MLL   module  obtained from $\R$ by removing  $\Nets$.}, where   $\Nets = \{\N_1, \dots, \N_h\}$ is a set of Bayesian sub-nets.
	For $\iphi$ a  \groundinterpretation,
	~~$\phisem{\R} = \project{\BigFProd_{i:1}^n\phisem{\N_i}}{\At{\Delta}} $.
\end{lemma}
So in particular:
\begin{cor} Let  $\R:\Delta$  be a   \bpn, and $\iphi$ a \groundinterpretation.
	Assume that  $\R$ is a cut-net   $\D(\N_1:\Gamma_1, \dots, \N_h:\Gamma_h)$, where $\N_1, \dots, \N_h$ are \bpn's.
	Then
			$\phisem{\D(\N_1, \dots, \N_n)} = \project{\BigFProd_{i:1}^n\phisem{\N_i}}{\At{\Delta}} $
\end{cor}

\begin{theorem}[Turbo  PCoh interpretation]\label{cor:turbo}\label{thm:turbo}
	 Let  $\R:\Delta$  be   a  \bpn  in \emph{\factorized} form, and  $\iphi$  a \groundinterpretation.   $\phisem \R$ can be inductively computed as follows 
	\begin{itemize}
		\item $\phisem{\bbox} = \iphi(\bbox)$, for $\bbox\in \Ax(\R)$
		\item $\phisem{\D(\N_1, \dots, \N_n):\Delta} = \project{\BigFProd_{i:1}^n\phisem{\N_i}}{\At{\Delta}} $
	\end{itemize}
\end{theorem}

\begin{example}[Factorized Interpretation]\label{ex:fact_2}
	Let us compute the interpretation of   \bpn $\R_D'$  in 
	\Cref{fig:ABCDE4}, with  the three modules completed as in  \Cref{ex:fact}.
	With the same notations, let $\N_1 = \D_1(\bbox^A, \bbox^B,\bbox^C)$,
	 $\N_2= \D_2(\bbox^D,\N_1)$, $\N_3=\D_3\big(\bbox^E, \N_1\big)$.
	 Notice  that  $\D_1$ has  output $ B^+,C^+ $, while 
	  $\N_2$ and $\N_3$ have output  $  C^+,D^+$, and $ D^+$, respectively.
\begin{center}
{\small 		$ 	\begin{array}{c|c|c|c}
	\text{sub-net}		&  & \sem {\N_i}  &\Var{\D_i}  \\
			\hline
		\N_1	& \project{\phi^A\FProd \phi^B\FProd \phi^C }{B,C}& \sem {\N_1}=\tau^1(B,C)& ABC  \\
		\N_2	&  \project{\phi^D\FProd \tau^1(B,C) }{C,D}& \sem {\N_2}=\tau^2(C,D) & BCD\\
		\N_3	& \project{\phi^E\FProd \tau^2(C,D) }{D} &  \sem {\N_3}=\tau^3(D)& CDE
		\end{array} $}
\end{center}
When computing this way  the interpretation of  $\sem {\R_D'}$, 
the  size of the largest factor which is produced in the calculation is $2$, 
 to be contrasted with what happens for  
 ${\R_D}$ (\Cref{ex:roadmap}).
\condinc{}{Assuming the cardinality of $\Val X$ (for $ X\in \{A,B,C,D,E\} $) is at most $k$, we moved from  calculating and storing $\BigO{ k^4}$  values, to $\BigO{k^2}$.}
 \SLV{}{We stress that $\sem {\R_D} =  \sem {\R_D'}$. }

\end{example}

\subsection{Cost of   interpreting a \factorized net}


Given a \bpn $\R$, we define  the parameter $m_\R$ which counts the number 
of components in the cut-net $\R$ (minus one)
\begin{center}$
	 m_\R=\size{\Ax{(\R)}}~+~\size{\Wirings{\R}} - 1$
\end{center} 
\begin{theorem}[Cost of interpreting a \factorized \bpn]\label{prop:cost} Let $\R$ be a \bpn in \factorized form. 
	The time and space cost of computing the  interpretation $\phisem \R$ is 
\begin{center}
	$ 	\BigO{m_{\R} \cdot \Exp {\Width\R}} $
\end{center}
for  
$ \Width\R$  as in \Cref{def:width_R} and $m_{\R}$  as defined above.
\end{theorem}
\SLV{}{
\begin{proof}
First, observe that  $Ax{(\R)}\cup\Wirings{\R}$ are the components of the cut-net $\R$, and so the nodes of the correction graph $\tree_\R$. So 
 $m_\R$  counts the number of edges in $\tree_{\R}$ (the number of nodes minus one).

Let $ \Width\R=w $.
Assume $\R=\D{(\N_1:\Gamma_1, \dots, \N_h:\Gamma_h)}$.	Each   $\phi_i =\sem {\N_i}$ is a factor over $\At{\Gamma_i}\subseteq \At{\D}$, so  that  $\size{\At{\Gamma_i}} \leq  \size{\At{\D}}\leq (w +1)$.   The cost of computing 
$\BigFProd_{i:1}^h\phi_i$ is  $\BigO{h\cdot \Exp{w+1} }$. Since the cost of the projection is similar, the cost of computing $\sem{\D(\N_1, \dots, \N_h)}$ is  $\BigO{h \cdot \Exp{w}}$ \emph{plus} the cost  of  computing 
 $\phi_1=\sem{\N_1}, \dots \phi_h=\sem{\N_h}$.


By \ih, for each $i$, the cost of computing   $\sem{\N_i}$ is $\BigO{m_{\N_i} \cdot \Exp {w}}$, so 
$\sem{\D({\N_1}, \dots, {\N_h})} = \BigO{(h + \sum_{i:1}^h m_{\N_i}) \cdot \Exp{w} }  $ 

Recall that   $m_{\R}$ is the number of edges in the tree  $\tree_{\R}$, which  has root $\D$, and a subtree $\tree_{\N_i}$ (with $m_{\N_i}$ edges) for each $\N_i$.
So the number of edges in $\tree_{\R}$ is  $m_{\R} = h+  \sum_{i:1}^h m_{\N_i}$. We conclude that  
 the total cost for computing $\sem{\R}$ 	is	\begin{center}
 	$ \BigO{m_{\R} \cdot \Exp {\Width\R}} $.
 \end{center}
\end{proof}
}

 In \Cref{sec:induced} we will give a procedure to rewrite  a \bpn $\R$ in a \emph{factorized} \bpn $\R'$ which has (possibly\footnote{We can produce all possible factorization (in a sense made precise by the notion of order), including those of minimal cost, but minimality is NP-hard.}) smaller width, 
 but  the same  boxes and the same conclusion as $\R$, and so the same interpretation. 
Moreover, we will have that  $m_{\R'}$ is  $ \BigO{n} $, for $n=\size{\At{\R}}$, so that  computing $\phisem{\R}$ can actually be achieved with cost  	 
\begin{align}\label{eq:cost}
\BigO{n \cdot \Exp {\Width\R}}
\end{align}

\section{Computing Marginals, efficiently}
Let us revisit the inference of marginals in terms of interpretation (\Cref{sec:PView}). 
For the sake of a compact  presentation,  we focus on a   marginal $\Pr_\R (Y)$ over a \emph{single variable} ---everything generalizes to  marginals $\Pr_\R (Y_1, \dots, Y_k)$ over \emph{several variables}, at the price of some  \emph{easy}  bureaucracy.
\begin{remark}[Why computing all single marginals matters] Given a \BN,  the task  of computing every single marginal,  is a relevant one in itself, needed for example by learning algorithms.  Software tools for BN's  typically\footnote{For example,  GeNIe Modeler by BayesFusion computes all marginals by  using the clique tree algorithm  which we recall  in \Cref{sec:MP}}
	 pre-compute each single prior marginal as  soon as the model is defined,   before any query. 
\end{remark}


Given a positive \bpn $\R$ and fixed a \groundinterpretation $ \iphi $, we can see (\Cref{sec:PView}) every  $X\in \At{\R}$  as a \rvar  whose associated set of values is $\iphi (X) = \Val X$.  
Then  $\R$ specifies  a joint probability $\Pr_{\R}$ over $\At{\R}$. 
Moreover, if  $\R$ has  conclusion $\Y$,  $ \phisem{\R:\Y}$ 
is exactly  the \emph{marginal} probability $\Pr_{\R}(Y)$. We wonder:
{``In the  proof-net setting, 
	how much work we need to compute each marginal $\Pr_{\R}(X)$, for $X\in \At{\R}$ ?''}

We show below  that once  we have  a factorized \bpn $ {\R:\Y}$ which allows us to compute $ \phisem{\R:\Y}$ at a certain cost, then 
 we can  compute  every other marginal on a (easy to obtain) variant of the same factorized \bpn, 
 at the same cost.

 \begin{remark}Given  a positive \bpn $\R$,
 $\OHide$ and $\OShow$ allow us to modify the conclusion of  $\R$ without altering  $\Pr_{\R}$, and so   compute the marginal probabilty for any $X\in \At{\R}$.  
 First, notice that  by repeated \emph{hiding}, we can always transform 
 a positive \bpn $\R$ into a \bpn $\R_0$ of empty conclusions (clearly, $\sem \R \not= \sem{\R_0}$!)  but  with $\Pr_{\R} = \Pr_{\R_0}$.
 Conversely, 
if  $\R$ has empty conclusions and   $Y\in   \At{\R}$ is an internal \atom of $\R$,  \emph{$\OShow$ } of
$Y$  yields a proof-net $\R'$ which has conclusions $\Y$.
By \Cref{prop:bpn_vs_bn},	$\phisem{\deW {\R} {Y}}$ is   the marginal probability of $\Pr_\R$ over ${Y}$.

 \end{remark}

\begin{remark}[Empty conclusions] 
	A \bpn of empty conclusion has an especially simple structure.  The previous remark enables  us to work with empty conclusions (a technique we  take from  \cite{locus}), 
	  easing  the work-load without any loss of information. 
\end{remark}

\subsubsection{Computing each single marginal.}
We use in an essential way   \Cref{lem:root_change}, relying on the fact that a \factorized \bpn  $\R$ is a \emph{rooted}  cut-net.  


Assume we have a factorized \bpn $\R_0$ of \emph{empty} conclusion,  width $w$, and internal atoms $X_1, \dots ,X_n$. We immediately have available---for each internal atom $X_i$---a factorized \bpn   $\R:\X_i$,  with  the \emph{same width} as that of
$\R_0$, namely $\R=\deW \R{X_i}$.  To realize this, 
recall that a factorized \bpn is simply a cut-net of components $\Ax{(\R)} \cup \Wirings{\R}$---we can choose as root $\D_r$ any of its \wirings (\Cref{lem:root_change}). For any  internal atom $X$ such that $\cut(\X,\Xn)$  appears in $\D_r$, 
 de-weakening  ($\OShow$, \Cref{def:de-weakening}) allows us to "bring-out" a positive edge $e:\X$,  making $\X$ visible in the conclusions.
 \begin{example}See \Cref{ex:root}.
 \end{example}
\begin{prop}\label{thm:any_marginal1} Let $\R_{}$ be a factorized \bpn of \emph{empty conclusions}. 
For each   $Y\in \At{\R}$,	 $\R_Y=\deW {\R_{}} Y$ 
		is a factorized  \bpn of  conclusion $\Y$. We have ${\At{\R_{}}}={\At{\R_Y}}$, $\Ax(\R_{})=\Ax(\R_Y)$, $\Width{\R_{}}=\Width{\R_Y}$,
		$m_{\R_{}}=m_{\R_Y}$.
\end{prop}
So  the cost of  interpreting $\R_Y$ is  always the same.
\begin{theorem}\label{thm:any_marginal2}
	Let  $\R$ be a \bpn of empty conclusion, and $\iphi$ a \groundinterpretation.
	For \emph{every} $Y\in \At{\R}$,  the cost of computing $\phisem{\R_Y:\Y}$  is   the same as the cost of computing  $\phisem{\R}$. 
\end{theorem}
\SLV{}{
\begin{proof}Recall that  $\R_Y=\deW{\R}{Y}$ is obtained either by removing a $\w$-node, or by adding an axiom $\ax(\Y,\Yn)$, where the negative edge $e:\Yn$ is premise of a $@$-node, and  the  positive edge $f:\Y$ is the conclusion of $\R_Y$. Let   $\D$ be  the \wiring in which the axiom has been added, and $\D'$ the corresponding wiring in  $\R_Y$ . Notice that $\Width \D = \Width{ \D'}$. 
	
	By \Cref{lem:root_change} , we can see  $\R_Y$ as  a factorized form which has $\D$ as root. The claim then follows from 
	\Cref{prop:cost}, and  the fact that $\OShow$  does not change the size of any \wiring. 
\end{proof}}
We  use  (a variant of)  the same \bpn for all the marginals.
%
\begin{corollary}\label{cor:marginal_cost}
	With the same assumptions as in \Cref{thm:any_marginal2}, assume  
	  $\size{\Wirings \R}$ is  $\BigO{n}$, for  $n=\size {\At{\R}}$. 
	For each $Y\in \At{\R}$, the   marginal $\Pr_\R(Y)$ can be computed with  cost  \begin{center}
		$	\BigO{n \cdot \Exp {\Width\R}}$.
	\end{center}
\end{corollary}

%

%% file: 04_VE_bpn.tex
\section{A cuts-led factorization}\label{sec:induced}
We have discussed the properties of factorized \bpn's. 
In this section, we sketch  an effective algorithm to factorize a positive \bpn $\R$ (the actual construction is  in  \Cref{app:factorize}).
Our goal is to break the  large wiring corresponding to $\nax \R$ (as in \Cref{fig:ABCDE3})
into smaller components (as in  \Cref{fig:ABCDE4}). Such smaller components  will be \emph{cut-free} \emph{\wirings}.
Concretely, 
we progressively rewrite a \emph{normal} \bpn $\R$ into a factorized \bpn $\R'$,  
by choosing and processing  \emph{one node  $\cut(\X,\Xn)$ at a time.}
These choices  eventually define a total order $\omega$ on the $\cut$-nodes and so (since there is exactly one cut per internal variable) on the internal atoms of $\R$.
  The rewriting only consists of $\arule$-rules expansions, hence $\sem {\R'}=\sem \R$ and  $\bnet{\R'}=\bnet{\R}$. 
Notice that $\R'$ is no longer normal, but \emph{its structure is simpler}\footnote{We borrow this technique from \cite{locus}}. Also notice that  $\R'\Red[]^*\R$.
We first  sketch the  idea.

 \paragraph{Naive idea: one $\cut$ at a time}

 The natural way to interpret  a \bpn $\R$ is to follow  the sequentialization of $\R$. If we adopt  a sequentialization algorithm which works top-down---starting  from the boxes\footnote{Parsing (\eg   \cite{GuerriniM01}) is a sequentialization algorithm which works this way}---we can  interpret the sequent calculus derivation while we are building it. 
This process, progressively interpreting the (virtually built) derivation tree,  already provides
 an instance of factorized interpretation. 
 
 Consider again $\R_D$ in \Cref{fig:ABCDE3}. The sequentialization\footnote{Sequentialization here  greedily performs \emph{contractions} whenever possible.} in \Cref{fig:ABCDE_seq} gives an inductive construction---we can follow it. 
 We  have the interpretation of the leaves (the boxes), and we    proceed inductively, according to  the "standard" semantics of PCoh. 
 In particular, a sub-tree ending with a  $\cut$-rule is interpreted by composing  the  interpretations of its two children.

\condinc{}{
 \[\infer[\cut (A)] {\quad\vdash D^+}
 {\infer[\bbox^A]{\vdash A^+}{}&
 	\infer[\cut(C) ]{\vdash A^-,D^+ }
 	{\infer[\bbox^C]{\vdash A^-,C^+}{} & \infer[\cut (B)]{\vdash A^-,C^-,D^+}{\infer[\bbox^B]{\vdash A^-, B^+}{} & \infer[\bbox^D]{\vdash B^-,C^-,D^+}{}}
 }
}
\]
}
 Notice how 
  the interpretation/sequentialization  resemble a  \emph{kind of cut-elimination} (very much like in \cite{locus}). In our example, we proceed according to  the order: \begin{center}
  	$\cut(E) , \cut(B), \cut(C), \cut(A)$.
  \end{center}
 The size of the  largest factor we produce in this  example  (being  a bit smart with  contractions and mix)  is $3$. For instance, when interpreting the proof of root $\cut(B)$, we have 
 $\tau^B(A,C,D) = \project{\phi^B\FProd \tau^E(B,C,D)}{A,C,D}$.
 
 The limit of this approach is that  the factorization directly obtained by   sequentialization  is not necessarily the best one. 
 A better factorization for $\R_D$ would be as in \Cref{fig:ABCDE4}, yielding intermediate factors of maximal size $2$
 (\Cref{ex:fact_2}). Indeed, 
 we can reorganize (factorize) the original proof-net differently, \emph{effectively sequentializing 
 in a computationally more efficient way. }
 The \pn in \Cref{fig:ABCDE4}  corresponds to the  proof sketched below, once the $\otimes/\parr$-cuts have been reduced.
 
\begin{center}
	{\footnotesize
	 $\infer[\cut(B,C)]{\vdash D^+}{
 \infer[\cut (A)]{\vdash B^+\otimes C^+}{  {\vdash A^+}&   \infer{\vdash A^-,B^+\otimes C^+}{\vdash A^-,B^+&   A^-,C^+}}
 &   \infer{\vdash C^-\parr B^-}{\infer {\vdash C^-,B^-,D^+}{\vdots}}
 }$
 }
\end{center}
We need not (and will not) explicitly go via $\otimes/\parr$-cuts. Instead we will produce  \emph{bundles of cuts}, where each bundle  could actually be $\otimes/\parr$-expanded into a single cut.

\subsection{Factorization induced by  Cut-Elimination order}\label{sec:induced_factorization}
 To simplify the proofs   (and without loss of generality) we assume the \bpn $\R$ to be   \emph{normal} and   of \emph{empty conclusions}. 
\begin{lemma}[Cuts] If a \bpn $\R$ has \emph{empty conclusion}, and is \emph{normal},  then for each atom $X\in \Var{\R}$ (1.) $X$ is \emph{internal},  and (2.)  there \emph{exists  
	unique}  a node $\cut(\X,\Xn)$.
\end{lemma}
Given a \bpn $\R$ we can produce  a factorized \bpn $\R_\omega$ which satisfies  the assumptions necessary to apply  \Cref{eq:cost} and \Cref{cor:marginal_cost}. 
The details of the construction are in \Cref{app:factorize}
\begin{theorem}[Factorized form of  $\R$ induced  by $\omega$]\label{thm:factorization} 
	Let  $\R$ be a normal \bpn of empty conclusion and \emph{$n$ internal atoms}. A total order   $\omega$   on   
$\At{\R}$ 
	induces a factorized \bpn $\R_\omega$ which is obtained from $\R$ by $\{\assoc,\ax \}$-expansions, so (for any \groundinterpretation $\iphi$)
	$\phisem {\R_\omega} = \phisem {\R}$,  and  $\bnet{\R_\omega} = \bnet {\R}$.
		$\R_\omega$ is a cut-net of \emph{at most $\mathbf{2 n}$ components}.
\end{theorem}	
The construction producing the  factorized \bpn 	$\R_\omega$ is given in \Cref{sec:factorization_algo}.  	$\R_\omega$  is a cut-net  of components $\Ax(\R_\omega) \cup \Wirings{\R_\omega}$, where $\Ax{(\R_\omega)}=\Ax(\R)$
	and each $\D\in  \Wirings{\R_\omega} $ is   \emph{cut-free}.

%% file: 04_VE_algo.tex


\section{BN's and Proof-Nets:  Inference}\label{sec:MP}

We conclude the paper by making explicit  the  tight correspondence between   factorized  interpretation of a \bpn, and  {exact} Bayesian   inference.

\subsubsection{Message Passing and Cliques Trees, in a nutshell}
\condinc{}{We first recall two  main such algorithms ---Variable Elimination and Message Passing.
 The  key insight  is to exploit the distributivity of sum (marginalization) over multiplication (product) to decompose the problem of inferring a marginal into smaller problems.
}

 The most widely used \emph{exact} inference algorithm is \emph{message passing over a clique tree}. 
Inference relies here on the compilation of a \BN into a data structure called a \emph{clique tree}. 
We recall some basic facts, and  refer \eg to \cite{DarwicheBook} (Ch.7) or \cite{KollerBook} (Ch.10) for details.
\SLV{}{\footnote{For the reader convenience, in \Cref{app:MP} we give a  brief summary, and  the  basic  (single traversal) Message Passing algorithm.}. }
 At  high-level, the algorithm partitions the graph into clusters  of variables (aka cliques). The interactions among cliques  have a tree structure. Inference of marginals is then performed by message passing over the  tree (essentially,  a traversal of the tree).

\condinc{}{At a high-level the clique tree algorithm partitions the graph into clusters of variables. The dependencies among variables can be quite complex, however, interactions among clusters will have a tree structure. 
}

\begin{definition}[Clique Trees]\label{def:clique_tree}
	Let  $\netB=(\netG, \aphi) $  be a \BN  over variables $\bX$. 
	A clique tree   for  $\netB$ (more precisely, for its DAG)
	is a pair $(\tree, \cC)$ where $\tree$ is a tree and $\cC$ is a function that maps each node $i$ of the tree $\tree$ into a subset $\cC_i \subseteq \bX$, called a \emph{clique}. Moreover,  for each $X\in \bX$:
	\begin{itemize}
		\item $\{X\}\cup \Pa{X} \subseteq \cC_i $, for some  $\cC_i$. \emph{(Family preservation)}. 
		\item  If $X$ appears in two cliques $\cC_i$ and $\cC_j$, it must appear in every clique  $\cC_k$ on the path connecting nodes $i$ and $j$. (\emph{Jointree property}).
	\end{itemize}
	The \textbf{width} of a clique tree is defined as the cardinality  of its largest clique  minus one.
\end{definition}
Assume  $\netB$ is a BN over $n$ variables $\bX$, and $(\tree, \cC)$ a clique tree for $\netB$ which has  $\BigO{n}$ nodes and width $w$. Then   for any $X\in \bX$, the message passing algorithm computes $\Pr(X)$, with a single traversal, and cost  $ \BigO{n~ \Exp{w}}$. 
\begin{remark}
It should be noted (but  we do not go so far here)  that by  exploiting the properties of clique trees,  message passing can   compute the marginal for  \emph{all  the $n$ variables} in the original \BN, with  only  two traversals of the  tree,
 at a \emph{total} cost  of  $ \BigO{n~ \Exp{w}} $.
\end{remark}
 Several  algorithms are available  to {compile} a BN into a clique tree. A way is  to generate  a clique tree 
  $\Cliques \netB \omega$  from a total  order $\omega$ over $\bX$. Such  $\omega$ is called an elimination order, which is  an important data structure in its  own right, with a specific notion of width\footnote{All width notions    are lower-bounded by  an intrinsic parameter of the BN.}. 
Poly-time algorithms  convert between cliques trees and elimination orders, while preserving the respective width, hence, the cost of computation. We stress that to compute  an order or a clique tree of \emph{minimal width} is  an NP-hard problem (with good heuristics).


\subsection{A proof-theoretical counter-part}\label{sec:pt}
We   make explicit  the correspondence between   factorized proof-nets and clique trees, and so the correspondence between    computing the (factorized) interpretation of a cut-net   interpretation (\Cref{thm:turbo})  and  
inference by a single traversal of message passing on the corresponding clique tree.

\subsubsection{Factorized proof-nets as clique trees}
Recall that to a cut-net (so in particular to a factorized \bpn) is   associated  a (correction)  tree (\Cref{def:cut-structure}).
\newcommand{\patht}{\mathfrak t}
\begin{proposition}[Clique tree] \label{fact:tree} Let $\R$ be a   factorized \bpn, and $\tree$  the restriction of its correction graph  $\tree_\R$ to $\Wirings{\R}=\{\D_1, \dots, \D_n\}$. 
	The tree $\tree$   equipped with the function $\cC$ which  maps each $\D_i$ into  $\At{\D_i}$ is a \emph{ clique tree}.
\end{proposition}
\begin{proof}
The key property to prove is the jointree property, which follows from the \goodness of the \MLL module  $\nax{\R}$
(see Appendix).
\end{proof}
\noindent

\SLV{}{
	All properties in \Cref{def:clique_tree} are easy to verify.
	The jointree property is  given by the following Lemma.
	
	\begin{lemma}Let $\tree(\R)$ be as in \Cref{fact:tree}.
		If $X\in \At{\D_i}$ and $X\in \At{\D_j}$, then $X\in \At{\D_k}$, for each  $\D_k$ which is in the path between $\D_i$ and $\D_j$ in $\tree(\R)$.
	\end{lemma}
	
	\begin{proof}
		Notice that $\tree$ is a tree  which has for nodes $\Wirings{\R}=$ and an edge  $(\D_i, \D_j)$ exactly when there is a cut between $\D_i$ and $\D_j$.
		
		The claim follows from the fact that $\nax{\R}$ is  \good (\Cref{lem:B_connected}  	).
		Let $e_i$ and $e_j$ be two edges supported by $X$, with  $e_i$ in  $\D_i$  and $e_j$
		in $\D_j$.  Since $\nax{\R}$ is \good,   in $\R$ there is path $\mathfrak r$ between $e_i$ and $e_j$ whose edges are all supported by $X$. Such a path 
		is necessarily reflected into a path $\patht $ in $\T_{\R}$.

		Since  $\T_{\R}$ is a tree, there is only one path between $ \D_i $ and  $\D_j$. We prove (2.) by induction on the path $\patht$. 
		If  $i=j$, the claim is trivially true. Otherwise,  let $\patht=\D_i,\D_h, \dots, \D_j$. 
		Necessarily, in the cut-net  $\R$ there is a  $\cut(\X,\Xn)$ connecting  the net $\D_i$ and  the net $\D_h$, and so $X\in \At{\D_h}$.
		We conclude by \ih.
		
		
	\end{proof}
}
\condinc{}{
	\begin{proposition}[A factorized \bpn  as a clique tree]  Let $\R$ be a \factorized \bpn. Let $\tree$ be the restriction of the tree $\tree(\R)$ to $\Wirings{\R}=\{\D_1, \dots, \D_n\}$,
		as in \Cref{fact:tree}. 
		The tree $\tree$   equipped with the function $\cC$ which  maps each $\D_i$ into  $\At{\D_i}$ is a \emph{ clique tree} (\Cref{def:clique_tree}).
		\SLV{}{	\begin{enumerate}
				
				\item The tree $\tree_{\R_\omega}$, equipped with the function $\bC$ which  maps  $\D_i$ into  $\Var{\D_i}$ is a\textbf{ clique tree}.
				
				\item If the  \wiring $\D_j$ has an input $\Delta_i$ which is  conclusion of a \wiring $\D_j$, then  $\bC_i \cap\bC_j  = \At{\Delta_i}$. 
				
			\end{enumerate}
		}
	\end{proposition}
	
}

\vspace*{-4pt}
\subsubsection{PCoh interpretation and Message Passing}
Given the above correspondence, 
one can easily see  that the  interpretation  by \Cref{thm:turbo}, and the message passing algorithm  (recalled in \Cref{app:MP})) on the corresponding clique tree behave similarly, computationally.  Let us be more precise.

Let $\R $ be a normal \bpn  of empty conclusion.
Let   $\omega $  be  a complete order on $\At{\R}=\{ X_1, \dots,X_n\}$, and  $\R_\omega$  the factorized form of $\R$ induced by $\omega$. 
Fixed a \groundinterpretation  $\iphi$:
\begin{itemize}
	\item to  $\R$ and  $\R_\omega$ correspond  the same BN  $\netB= \Bnet{\R,\iphi} = \Bnet{\R_\omega,\iphi}$
	and the same probability $\Pr_{\R}=\Pr_{\R_\omega}$.
	\item  for any arbitrary $Y\in \At{\R}$,  
	the  cost of computing $\Pr_\R(Y) $ via factorization in $\R_\omega$ (\ie  the cost of computing $  \phisem {\N_Y} $,
	for $\N_Y=\deW{\R_\omega}{Y}$) is
\[\BigO {n \cdot \Exp {\Width{\R_\omega}}}\]
\end{itemize}
Let us compare with (single traversal) message passing over the clique tree $\Cliques \netB \omega$  induced by $\omega$
on $\netB= \Bnet{\R,\iphi}$. 
	Its width is indeed an upper bound for $\Width{\R_\omega}$, so  that the  cost is of the same order (details in \Cref{app:Binference}).
\begin{prop}\label{prop:cost} With the same assumptions as above, let $w$ be  the width of the clique tree 
	$\Cliques \netB \omega$  induced by $\omega$, then:
	\begin{center}
		$\Width{\R_\omega}$ is no greater than   $w$.
	\end{center}
 
	%
\end{prop}
\SLV{}{\begin{proof}$\R_\omega$ has a similar structure to that of $\tree(\netB, \omega)$, and so a similar width. A formal proof of this result is  straightforward. 
\end{proof}
}

		%
		%
		

\condinc{}{
		\begin{corollary}[Optimality] Let $\R$,  $\omega$, $w$,  $\iphi$, 
			$\netB$ as stipulated above. Let 
			$\R_\omega$ be the factorized form of $\R$ induced by $\omega$. 
			
			 Given  $Y\in \At{\R}$,  let $ \R_Y=\deW{\R_\omega}{Y}$ 
			The  time and space cost of computing the interpretation $\phisem {\R_Y}$ is $\BigO{n \cdot \Exp {w}}$, 
			which is 	the same  cost as  computing  $\Pr_{\netB}(Y)$ by running \Cref{alg:messages} over the  clique tree 
			$\tree(\netB, \omega)$ (\Cref{fact:mpVE}).
		\end{corollary}
}



\condinc{}{

\paragraph{Discussion}
The cost   for the interpretation does not take into account the preliminary  factorization of $\R$. Notice however that this transformation (which is only concerned with the proof-net, not with  the quantitative information) needs only to be performed once---it is comparable to the compilation of a \BN into a clique tree. 

Cost-wise, the  critical aspect in   the factorization process is  the quality of the factorization (its width), which depends on the choice of the order. Again the problem is analogous to computing a good (ideally, optimal) elimination order---we  rely on the same algorithms. Agai nnotice that---like for message passing on clique trees---to produce a factorized form we only need a single order. Inference of  all the marginals can be   performed on the \emph{same} factorized \bpn. 

Furthermore, borrowing from message passing, we can adapt turbo interpretation in order to compute  all marginals $\sem{\R_Y}$, for all $X\in \bX$,
with a total cost that is only twice that of computing a single marginal.}

\begin{remark}
Borrowing from message passing, we could now    adapt \Cref{thm:turbo} so to compute  all marginals $\phisem{\R_X}$, for  $X\in \bX$,
with a total cost  only twice that of computing a single marginal.
\end{remark}

%% file: 07_Conclusions.tex

\section{Conclusions}
In this paper, we have developed a  \emph{proof-theoretical account} of Bayesian inference.
%
Linear Logic proof-nets are a graph representation of proofs and of their dynamics (cut-elimination), and ---via the Curry–Howard correspondence between  proofs and  programs---a graph representation of programs and of their execution.
Multiplicative proof-nets equipped with their PCoh interpretation  appear as a uniform setting in which is  possible both to represent probability distributions, and to perform inference---by using as data   structure the proof-net itself.
The following table summarizes  the   relationships  which we uncovered. 
\SLV{}{The following table summarizes  the relations  which we uncovered in this paper, somehow in the spirit
 of the 
Curry–Howard correspondence  between  proofs and  programs,  which involves both  the representation (proof-as-program, formula-as types) and  also---crucially---the computation (cut-elimination in proof theory and execution for programs).}

{\footnotesize \begin{center}
	\begin{tabular}{|c|c|}
	\hline
\MLLAx Proof-Nets $\R$	& Bayesian Structures $\netG$\\
\hline 
PCoh \groundinterpretation   $\iphi$ &  assignment of CPT's $\aphi$ \\
&  (Conditional Probabilities)\\
\hline
$(\R,\iphi)$  & Bayesian Network$  (\netG,\aphi) $ \\
\hline 
PCoh Interpretation	& Inference  \\
\hline
Conclusions of $\R$	&  Marginals \\
	\hline
Cut-Net	&  Clique Tree  \\
	\hline
Factorized   Interpretation  & Message Passing 	\\
of a   cut-net  & over a clique tree\\
\hline
\goodness	& jointree property \\
\hline
\end{tabular}
\end{center}
}

%% file: 99_Appendix.tex

\section*{APPENDIX}

\section{Multiplicative Linear Logic with Boxes}

We point out some  properties which  we use in the proofs.
\begin{remark}\label{rk:subnet}
	Notice that  a sub-graph of a switching acyclic module is  switching acyclic. 
\end{remark}

\begin{remark}\label{rk:normal_atomic}
	If 	  $\ProofN$ is a   normal \pn with atomic conclusions, then $\ProofN$ is  atomic.
\end{remark}

\begin{property}[normal proof-nets ]\label{lem:normal_net}
	Let  $\R:\Delta$ be a  \emph{normal} \pn. Then 
	(1.)  the premises of each $\cut$-node are atomic, 
	(2.) the  positive premise  of each cut-node is  conclusion of a $\Ax$-node. 
\end{property}

\condinc{}{
	\begin{lemma*}[\ref{lem:pol_correct} Polarized correctness] For $\M$  an atomic  \emph{raw} module, we denote by  $\pol{\M}$  
		the  graph which has the same nodes and edges  as $\M$,  but where the edges are directed 
		\SLV{}{according to their polarity:} downward if positive,  upwards if negative. 	
		The following are equivalent: (1.)
		$\M$ is acyclic correct, (2.)
		$\pol \M$   is a DAG.
	\end{lemma*}		
	
	{\begin{proof}We only prove $2. \Rightarrow 1.$ (the other direction is immediate). Recall that a DAG where each node has at most an outcoming edge is a tree. By inspecting the grammar of nodes, we observe that, w.r.t. the \emph{polarized orientation},
			only $\parr$ and $@$ nodes may have more than one outcoming    edges. Hence each switching graph is a tree. 
	\end{proof}}
}

\subsection{Sequent Calculus for \MLLAx}\label{app:sequent}
\Cref{fig:sequentialization} gives the rules to generate the sequent calculus derivations for \MLLAx. It also gives, for each derivation 
 $ \pi $ its  image $\N$ as a proof-net. Reading the figure right to left, we have a   sequentialization of the proof-net $\N$.
\begin{figure}
	\centering
	\includegraphics[page=5,width=1\linewidth]{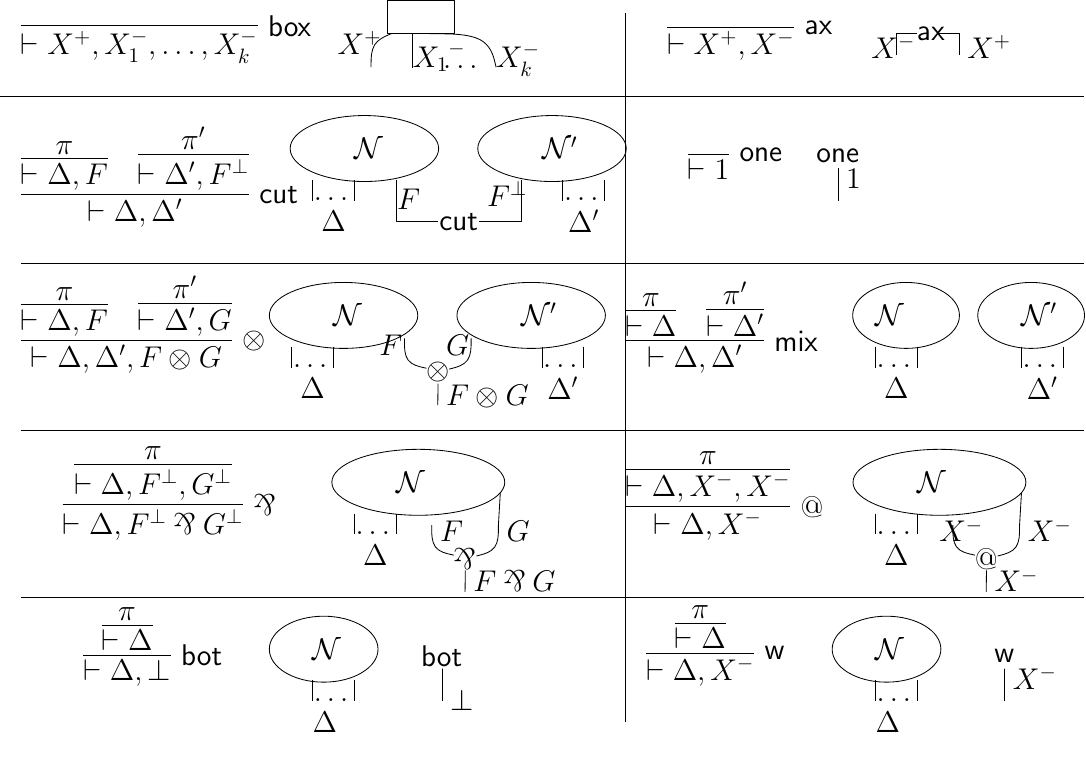}
	\caption{Sequent Calculus, Proof-Nets, and Sequentialization}
	\label{fig:sequentialization}
\end{figure}

\condinc{}{
$ \infer[\mathsf{box}]{\vdash\X,\Xn_1,\dots,\Xn_k}{} $

$ \infer[\ax]{\vdash \X,\Xn}{} $

$ \infer[\mathsf{one}]{\vdash \One}{} $

$\infer[\mathsf{cut}]{\vdash \Delta, \Delta'}{\infer{\vdash  \Delta, F}{\pi}& \infer{\vdash \Delta',F^\bot}{\pi'}} $

$\infer[\otimes]{\vdash \Delta, \Delta', F\otimes G}{\infer{\vdash  \Delta,F}{\pi}& \infer{\vdash \Delta',G}{\pi'}}  $

$\infer[\mathsf{bot}]{\vdash \Delta,\bot}{\infer{\vdash\Delta}{\pi}}$

$\infer[\w]{\vdash \Delta,\Xn}{\infer{\vdash\Delta}{\pi}}$

$\infer[\parr]{\vdash \Delta,F^\bot\parr G^\bot}{\infer{\vdash\Delta, F^\bot,G^\bot}{\pi}}$

$\infer[\cn]{\vdash \Delta,\Xn}{\infer{\vdash\Delta, \Xn , \Xn }{\pi}}$

$\infer[\mathsf{mix}]{\vdash \Delta, \Delta'}{\infer{\vdash  \Delta,}{\pi}& \infer{\vdash \Delta'\bot}{\pi'}} $

}

\section{\MLL  connected components (\Cref{sec:MLLcomponents}) }\label{app:MLL}

\begin{lemma*}[\ref{lem:connected}] Let $\M$ be an atomic \MLL module  which is a connected graph. Then:
	(1.) 	All edges are \supported by the same \atom. 
	(2.) $\pol \M$---\ie $\M$ with the polarized orientation---is a directed tree, with a \emph{unique} initial edge.
\end{lemma*}
{\begin{proof}All properties follows from the fact that  an \MLL module only contains  nodes of sort $\ax,@,\weak,\cut$.   (1.) is immediate. For  (2.),  recall  that $\M$ with the polarized orientation is a DAG. Since all nodes $\ax,@,\weak$ have exactly one incoming edge (w.r.t. the polarized orientation), $\pol \M$ is necessarily a  tree. 
	\end{proof}
}

\begin{prop*}[\ref{lem:M_connected}  	]
	Let $\M: \Gamma \vdash \Delta$ be an atomic \MLL module (of conclusion $\Delta$,  premises $\Gamma$). Then  $\At{\M} = \At{\Gamma} \cup \At{\Delta}$, and 
	the following are \emph{equivalent} for each  $X\in \At{\M}$:
	\textbf{(1.) }	 $\M$ has exactly one initial edge supported by $X$. 
	\textbf{	(2.)} $\M$ has exactly  one maximal connected component supported by $X$.
	
\end{prop*}
{	\begin{proof} \SLV{}{ Observe that  an edge $e$ is initial in $\M$ if and only if it is initial in the maximal connected component to which it belongs.  }
		
		$ 1. \Rightarrow 2.$ Assume that for some \atom $X$, $\M$ has  two distinct maximal connected components $\M_X'$, with initial edge $e':\Xn$ and $\M_X''$, with initial edge $e'':\Xn$.   Necessarily, $e',e''$ are initial edges of $\M$, which contradicts  $(1).$

		$2. \Rightarrow 1.$ 
		Let $\M_X$ be the unique  maximal connected component in $\M$ supported by   $X$.  
		The fact that  $\pol{\M_X}$ has  a unique initial edge supported by $X$, implies that $\M$ has.
	\end{proof}
}

\section{Factorized form of a proof-net (\Cref{sec:factorized_form})}

{
	\begin{lemma}[normal connected components]	An atomic \MLL proof-net  $\D[X]$ which is connected 
		is  either a single  axiom, or a single    $w$-node with its conclusion, or a single  $@$-node together with its  conclusion and its premises, where  each premise completed by an axiom.

	\end{lemma}
	\begin{proof} By inspecting the rules.
		Notice that in a \emph{normal} atomic \MLL proof-net, the positive conclusion of an $\ax$ node is necessarily conclusion of the net, because it cannot be premise of a cut (that would create a redex).
	\end{proof}
}

{
	\begin{remark}
		$\D$  is a \wiring exactly when it is a justaposition of    connected   proof-nets, each supported by a distinct \atom. We then can write $\D=\biguplus_{X\in \At{\D}} \D[X]$. 
		
		Observe that  each connected net $\D[X]$ has exactly one negative conclusion (this is obvious when $\D[X]$ is normal).
	\end{remark}
}

\begin{lemma*}[\ref{lem:trivial}, Trivial factorized form]
	Every \bpn   $\R$ can be transformed by $\ax$-expansion into $\R'=\D(\Ax(\R))$, where $\D$ is the completion of ${{\nax \R}}$. 
\end{lemma*}	

	\begin{proof}
		Every   proof-net $\R$  can be decomposed as $\pair {\Ax(\R)} {\nax \R}$, and transofrmed by  $\ax$-expansion, 
		into $\R'=\D(\Ax(\R))$, where $\D$ is the completion of   $\nax \R$.

		$\D$ is  \good  because 	 $ \nax \R  $ is (by \Cref{lem:B_connected} ). 
	\end{proof}
	
	\begin{remark}\label{rem:trivial}
		Notices that, if $\R$ has empty conclusion and is normal, to obtain a cut-net  $\R'=\D(\Ax(\R))$ it suffices to $\ax$-expand the negative  conclusions of the boxes, because every positive conclusion is already the premise of a $\cut$.
	\end{remark}

\begin{theorem}[Cost of interpreting a \factorized \bpn]\label{prop:cost} Let $\R$ be a \bpn in \factorized form. 
	The time and space cost of computing the  interpretation $\phisem \R$ is 
	\begin{center}
		$ 	\BigO{m_{\R} \cdot \Exp {\Width\R}} $
	\end{center}
	for  
	$ \Width\R$  as in \Cref{def:width_R} and $m_{\R}$  as defined above.
\end{theorem}

	\begin{proof}
		First, observe that  $Ax{(\R)}\cup\Wirings{\R}$ are the components of the cut-net $\R$, and so the nodes of the correction graph $\tree_\R$. So 
		$m_\R$  counts the number of edges in $\tree_{\R}$ (the number of nodes minus one).
		
		Let $ \Width\R=w $.
		Assume $\R=\D{(\N_1:\Gamma_1, \dots, \N_h:\Gamma_h)}$.	Each   $\phi_i =\sem {\N_i}$ is a factor over $\At{\Gamma_i}\subseteq \At{\D}$, so  that  $\size{\At{\Gamma_i}} \leq  \size{\At{\D}}\leq (w +1)$.   The cost of computing 
		$\BigFProd_{i:1}^h\phi_i$ is  $\BigO{h\cdot \Exp{w+1} }$. Since the cost of the projection is similar, the cost of computing $\sem{\D(\N_1, \dots, \N_h)}$ is  $\BigO{h \cdot \Exp{w}}$ \emph{plus} the cost  of  computing 
		$\phi_1=\sem{\N_1}, \dots \phi_h=\sem{\N_h}$.
		%
	%
		By \ih, for each $i$, the cost of computing   $\sem{\N_i}$ is\\ $\BigO{m_{\N_i} \cdot \Exp {w}}$, so\\
		$\sem{\D({\N_1}, \dots, {\N_h})} =$ $ \BigO{(h + \sum_{i:1}^h m_{\N_i}) \cdot \Exp{w} }  $ 
		
		Recall that   $m_{\R}$ is the number of edges in the tree  $\tree_{\R}$, which  has root $\D$, and a subtree $\tree_{\N_i}$ (with $m_{\N_i}$ edges) for each $\N_i$.
		So the number of edges in $\tree_{\R}$ is  $m_{\R} = h+  \sum_{i:1}^h m_{\N_i}$. We conclude that  
		the total cost for computing $\sem{\R}$ 	is	\begin{center}
			$ \BigO{m_{\R} \cdot \Exp {\Width\R}} $.
		\end{center}
	\end{proof}

\begin{thm*}[\ref{thm:any_marginal2}]
	Let  $\R$ be a \bpn of empty conclusion, and $\iphi$ a \groundinterpretation.
	For \emph{every} $Y\in \At{\R}$,  the cost of computing $\phisem{\R_Y:\Y}$  is   the same as the cost of computing  $\phisem{\R}$. 
\end{thm*}

	\begin{proof}Recall that  $\R_Y=\deW{\R}{Y}$ is obtained either by removing a $\w$-node, or by adding an axiom $\ax(\Y,\Yn)$, where the negative edge $e:\Yn$ is premise of a $@$-node, and  the  positive edge $f:\Y$ is the conclusion of $\R_Y$. Let   $\D$ be  the \wiring in which the axiom has been added, and $\D'$ the corresponding wiring in  $\R_Y$ . Notice that $\Width \D = \Width{ \D'}$. 
		
		By \Cref{lem:root_change} , we can see  $\R_Y$ as  a factorized form which has $\D$ as root. The claim then follows from 
		\Cref{prop:cost}, and  the fact that $\OShow$  does not change the size of any \wiring. 
\end{proof}

%

%% file: 99_Factorization.tex


\section{Factorized form induced  by a Cut--Elimination order (\Cref{sec:induced_factorization})}\label{app:factorize}

We  prove  \Cref{thm:factorization}, by  providing an algorithm (\Cref{prop:induced}) to produce, from a \bpn $ \R $ in normal form and of empty conclusions,
the  factorized \bpn $\R_\omega$  (\Cref{def:induced}).

We first take a closer look to the structure of a module in normal form, and of empty conclusion.

\subsection{\MLL preliminaries}
Recall that we call \emph{wiring  module} (\Cref{def:wiring})   an \MLL module 
which  is \emph{\good}.

\begin{lemma}[empty conclusions] Let $\M$ be a \MLL atomic  module of empty conclusions.  Let $\Gamma$ be the sequent of pending premises of $\M$.
	Then
	\begin{enumerate}
		\item $\At{\M} =\At{\Gamma}$
		\item The following are equivalent
		\begin{itemize}
			\item $\M$ is \good (\ie, distinct maximal connected components are supported by distinct \atoms)
			\item distinct \emph{positive} edges in $\Gamma$ are labelled by distinct atomic formulas.
		\end{itemize}
	\end{enumerate}
	
\end{lemma}

\condinc{}{
	\begin{definition}
		We call  \contree (labelled by $\Xn$) a negative module which consists of either a single  negative edge ($e:\Xn$)  or a single  $@$-node together with its incident edges, or a single  $w$-node with its conclusion. 
		
		Observe that an \monowired net  in normal form is  the completion of 
		a \contree.
	\end{definition}
}

\newcommand{\Bpair}[2]{\langle{\langle#1,#2\rangle} \rangle}

\begin{lemma}[Normal wiring modules] Let $\M$ be an \MLL   module of empty conclusions which is \emph{normal} and \emph{\good}. 
	Let  $\Gamma$ be the sequent of its pending premises. 
	\begin{itemize}
		\item   For each $X\in \At{\M}$, there is exactly one  pending premise  labelled by $\X$, and therefore in $\M$ there is  exactly one node $\cut(\X,\Xn)$.
		
		\item Let $\M[X]$ be the  maximal connected   component  supported by $X$.  It contains the node
		$\cut(\X,\Xn)$.  The positive edge $f:\X$ belongs to $\Gamma$. For the negative edge $e:\Xn$,   either
		(i.) $e \in \Gamma$, or (ii.) $e$ is conclusion of a $\w$-node, or , or (iii)
		$e$ is conclusion of a $@$-node , whose premises all belongs to $\Gamma$.
		
		
	\end{itemize}
\end{lemma}

\begin{figure}
	\centering
	\includegraphics[page=7,width=0.8\linewidth]{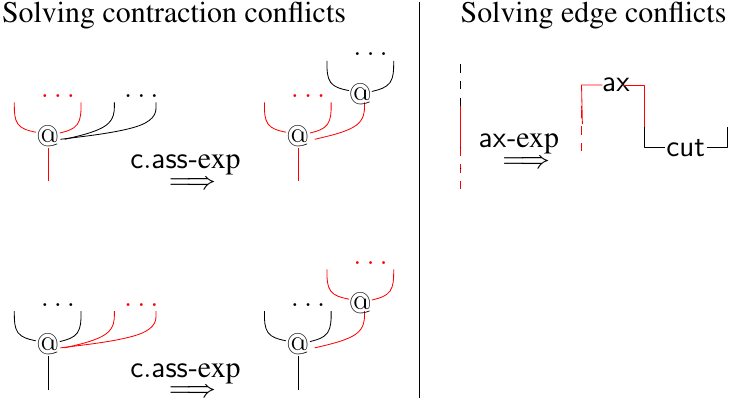}
	\caption{Please notice: this Figure is meaningful only in color!}
	\label{fig:solving}
\end{figure}

\subsection{Factorization Procedure}\label{sec:factorization_algo}
We prove \Cref{thm:factorization}. 
We factorize a $\bpn$ $\R$, rewriting  it (my means of $\ax$-expansions and $\assoc$-expansions) into a \bpn $\R_\omega$ (\Cref{def:induced}) which is in factorized form  (according to \Cref{def:factorized}), and has the properties stated in \Cref{thm:factorization}. 

\Cref{def:transformation} formalizes  the basic transformation which corresponds   to processing a single variable (a single $\cut$, as  described below).  We then 
\emph{iterate} the transformation, as stated in \Cref{prop:induced}, until we are left with a \bpn in factorized form. 
Notice that 
every application of the  basic transformation  preserves an invariant (\Cref{prop:induced}) which ultimately guarantees the correctness of the procedure.\\

The algorithm  is better understood by first considering   the decomposition of $\R$ as $\pair {\Box(\R)}\M$, where $\M=\nax \R$. Notice that $\Nets=\Box(\R)$ is  \emph{ a set of factorized nets}. 
Let us  now choose an arbitrary internal \atom $Z$---necessarily there is a $\cut(\Z,\Zn)$, connecting   the  unique box $\bbox^{Z}$ of conclusion $\vdash \Sigma, \Z$ with a number of boxes $\bbox_1,...\bbox_h$ of respective conclusions 
$\vdash \Sigma_i, \Zn$.  
The goal  now is to select   the minimal  sub-graph $\M^Z$ of $\M$ which contains  $\cut(\Z,\Zn)$ and  enough  structure   to compose together $\bbox^{Z}$ and $\bbox_1,...\bbox_h$. This oblige us to possibly include some other cut nodes (and so some other atoms).

 We then  expanded the sub-graph $\M^Z$ in order to obtain  a \wiring $\D^Z$, which wires  together $ \bbox^{Z},\bbox_1,...\bbox_h $. So we obtain  $\N'=\D^Z(\bbox^{Z},\bbox_1,...\bbox_h)$.
Crucially, observe that $\N'$ is a factorized net, and so 
the set $\{\N'\} \cup ( \Nets- \{\bbox^{Z},\bbox_1,...\bbox_h\})$ is a set of factorized nets. 
What is left  of $\M$ once we have subtracted $\M^Z$ is a \emph{smaller} \MLL module $M'$ in  which $Z$ does not appear anymore. We can now  iterate the process.

\condinc{}{
	Since $Z$   does not appear in the conclusion of $\R$, 
	necessarily there is a $\cut(\Z,\Zn)$, connecting   a unique net $\N^{Z}$ of conclusion $\vdash \Sigma, \Z$ with a number of nets $\N_1,...\N_h$ of respective conclusion
	$\vdash \Sigma_i, \Zn$ (figure). 
	The idea is to isolate  the minimal  sub-graph of $\M$ which contains $\cut(\Z,\bZ)$, and the structure which is necessary  to compose together $\N^{Z}$ and $\N_1,...\N_h$. The subgraph is then expanded in order to become  a proof-net $\D$, which allows us to perform the composition  $\D(\N^{Z},\N_1,...\N_h)$.
	
	What is left outside $\M^Z$ is a \MLL module $M'$ from which $Z$ has been removed, and we can iterate the process.
}

\subsubsection{The basic transformation: definition and properties of $\pmb{\El Z \R}$}\label{def:transformation}

Assume	 $\R $ is  a  \bpn of 
empty conclusion which decomposes (see \Cref{not:module}) as $\R= \pair \Nets \M$ where:
\begin{enumerate}
	\item 	
	$\Nets$ is a set of   \bpn's;
	\item	$\M$    is a wiring  module (\Cref{def:wiring})   in normal form (\ie,  $\M$ is an \MLL module 
	which  is \emph{normal} and \emph{\good}).
\end{enumerate}
Given a \bpn $\R$ which satisfies the above properties,  we define the  operation $\El - \R$, which 
for any arbitrary choice of an \atom $Z\in \At{\M}$, produces   a \bpn  $\R'=\El Z \R$ (obtained from $\R$ by expansion steps)
which satisfies the  properties  stated in \Cref{prop:transformation}.

Let  $\Nets_Z= \{\N:\Gamma  \in \Nets \st Z\in \At{\Gamma}\}$ and $\Gamma_Z$ the sequence of all the edges which are conclusion of some $\N\in \Nets_Z$.
Recall that by  assumption, $\M$ is a disjoint union of connected components in normal form, where distinct  components are supported by distinct  \atoms.

Consider  $\Gamma_Z$. For each  $X\in \At{\Gamma_Z}$, we first select  a subgraph $(\NN_X,\EE_X)$ of $\M$, and then expand it to obtain a proof-net.
\begin{enumerate}[(i.)]
	\item \textbf{\emph{Selection.}} 
	Let 
	$\EE_X$ be   the set of  the  negative \emph{edges} in $\M$ such that
	\begin{itemize}
		\item $e:\Xn\in \Gamma_Z$ 
		\item   $e:\Xn$  is connected via  a  $\cut$-node  to   $f:\X\in \Gamma_Z$ 
		\item $e:\Xn$ is conclusion of a $@$-node whose premises all belong to $\Gamma_Z$.
	\end{itemize}
	
	Let  $\NN_X$ be   the set of all \emph{nodes}  in $\M$ whose premises and conclusions are entirely contained in $\EE_X$.
	
	\item \textbf{\emph{Solve contraction conflicts: \textit{$\assoc$-expansions}} }(see \Cref{fig:solving}. Edges in $\EE_X$ are
	red-colored). 
	By the  assumptions, $\M$  contains at most one  $@$-node whose edges are labeled by $\Xn$.  If that $@$-node has some premises---but not all---belonging to $\Gamma_Z$, the  proof-net   is expanded   to obtain two separate 
	$@$-nodes, where  one is  attributed to $\NN_X$ (in red in  \Cref{fig:solving}). The new edge is attributed to $\EE_X$.
	%
	
	\item  \textbf{\emph{Completion}: \textit{$\ax$-expansions}} (see \Cref{fig:solving}).
	Each  edge in $\EE_X$ whose source is not in $\NN_X$  is $\ax$ expanded. The new $\ax$-node and its conclusions are  attributed to $\NN_X$ and $\EE_X$, respectively.
\end{enumerate}

\begin{itemize}
	\item Let $\R'$ be the proof-net obtained from   $\R$ by the expansions   steps at points (ii) and (iii). 
\end{itemize}
Notice that neither the  sub-nets in $\Nets$ nor the (empty) conclusion of $\R$ have  been modified, so $\R'$ is a \bpn.
\begin{itemize}	
	\item For each $X\in \At{\Gamma_Z}$, let $\D[X]$ be the sub-net of $\R'$ which has edges $\EE_X$ and nodes $\NN_X$. 
	\item Let  $\D \defeq \biguplus \D[X]$, for $X\in \At{\Gamma_Z}$. Observe that $\At{\D}=\At{\Gamma_Z}$ includes $Z$ (and possibly  also other \atoms).
\end{itemize}
\vskip 4pt

By construction: 
(i) $\D$ is  a \good \MLL proof-net, , and  (ii) 
each  conclusion of each $\N\in \Nets_Z$ is connected via a cut to  a  conclusion of   $\D$.  
\begin{itemize}
	\item Let  $\N'$ be  the subnet $\D(\Nets_Z)$. Its conclusions $\Delta$ are  the conclusions of $\D$ which are not connected to any net in $ \Nets_Z$ (\ie, $\Delta$ is the \emph{output} of $\D$).
	\item $\Nets' = \{\D(\InTrees Z)\}   \cup (\Nets - \InTrees Z)$
	\item Let $\M'$ be the \MLL module obtained from $\R^Z$ by removing $\Nets'$
\end{itemize}

\begin{remark}\label{rem:decomposition}
 Observe  that $\El Z{\R}= \R'$ decomposes as $\pair {\Nets'} {\M'} $.
\end{remark}

\begin{lemma}\label{lem:transformation} With the same assumptions and notations as above, 
	the following properties hold.
	\begin{enumerate}
		\item The  output  $\Delta$ of $\D$ is \repfree, and so $\D(\Nets_Z)$ is a \bpn.
		\item[2.] The \MLL module $\M'$  is normal and \good.
		\item[3.]    $Z\not\in \At{\M'}$.
	\end{enumerate}
\end{lemma}

\begin{proposition} \label{prop:transformation} With the same assumptions and notations as above,   let  $\R=\pair \Nets \M$,   	$\R'=\El{Z}  {\R}$, and $\pair {\Nets'} {\M'} $ the decomposition of $\R'$  in \Cref{rem:decomposition}.
The following properties hold:
	\begin{enumerate}
		\item 	
		$\Nets'$ is a set of   \bpn's. If $\Nets$ is a set of factorized nets, so is $\Nets'$.
		\item	$\M'$    is a wiring  module   in normal form.
		\item  $\At{\M'}\subsetneq \At{\M}$ and  $Z\not\in \At{\M'}$.
	\end{enumerate}
\end{proposition}
\begin{proof}	
	$\D(\Nets_Z)$ is a  \bpn by (1.) in \Cref{lem:transformation}. If $\Nets$ is a  set of  factorized  nets, then  $\D(\Nets_Z)$ is \emph{factorized} because  $\Nets_Z\subseteq  \Nets$, which  is a set of  factorized  nets. Hence point (1.) follows. Points (2.) and (3.) follow from the corresponding points in \Cref{lem:transformation}.
\end{proof}

\begin{remark}
	Notice that for each ground interpretation,	$\sem {\R'} = \sem {\R}$,  and  $\bnet{\R'} = \bnet {\R}$, because the interpretation and $\bnet -$ are invariant by $\{\assoc, \ax\}$ expansion.
\end{remark}

\vskip 4pt

\subsubsection{Factorized form induced by an order $\omega$}
 Given $\R$ a normal \bpn of 
empty conclusion, we produce a factorized form by repeating the application of the basic transformation described above,  eventually processing (eliminating) all atoms in $\R$. The correction of the construction is insured by a set of invariants which are satisfied by  the trivial decomposition of $\R$, and then preserved by each transformation step. 
\begin{proposition}[Induced factorized form]\label{prop:induced} For  $\R$ a normal \bpn of 
	empty conclusion,  let    $\omega = X_1,... ,X_n$  be an arbitrary sequence  enumerating   the    \atoms  in $\At{\R}$.
	To each initial segment   $X_1, \dots, X_j$  is inductively associated a \bpn
	$ \R_\omega^j$ together with a decomposition, as follows.
	\begin{itemize}
		\item[j=0.] $\R_{\omega}^0 \defeq\R$, with decomposition  $\pair {\Ax(\R)} {\M}$  as in \Cref{lem:support}. 

		\item[j+1.]Assume by \ih that  $\R_{\omega}^{j}$ has decomposition $\pair {\Nets} \M$. Using the operation in  \Cref{def:transformation},
		we define
		%
		%
		
		\[
		\R_{\omega}^{j+1} \defeq 
		\left\{
		\begin{array}{l}
			\El Z{\R_{\omega}^{j}}   \quad \mbox{ if  }  X_{j+1} \in \M \\
			\R_{\omega}^{j}	\quad \mbox{ otherwise.}
		\end{array}
		\right.
		\]


		%
		%
		
	\end{itemize}

	For each $j\geq 0$, the decomposition of   $\R_{\omega}^{j}$   as $\pair \Nets \M$ satisfies  the  invariants:
	\begin{enumerate}
		\item[I1] $\Nets$ is a set of \factorized \bpn's. 
		\item[I2]  $\M$  is a wiring module in normal form.  
		\item[I3] No  atom in $X_1, \dots, X_j$ appears in $\At{\M}$. 
	\end{enumerate}

\end{proposition}

\begin{proof}
	Clearly, the decomposition $\pair {\Ax(\R)} {\nax{\R}}$  satisfies all invariants. 	
	Assume by \ih that the decomposition $\R^{j}=\pair \Nets \M$ satisfies the  invariants, and let $\R^{j+1}=\pair{ \Nets'} {\M'}$.
	Invariant I1, I2, and I3 follow from the corresponding points in \Cref{prop:transformation}.
\end{proof}

\begin{remark}\label{rem:induced}
	Observe   that  for  $\R_\omega^n =\pair \Nets \M$ we have that $\M$ is empty  (by invariant I3).
\end{remark}

\begin{definition}[\factorized form $\R_\omega$ of  $\R$ induced  by $\omega$]\label{def:induced} 
		Let  $\R$ be a normal \bpn of empty conclusion, and  $\omega$  a total order  on $\At{\R}$. We denote by $\R_\omega$ the \bpn   $\R_\omega^n $ as defined in \Cref{prop:induced} (see also  \Cref{rem:induced}). 
	We call $\R_\omega $ 
	the factorized form   of $\R$ induced by $\omega$. 
\end{definition}

	By  construction, $\R_\omega$ has the properties stated in \Cref{thm:factorization}.
	
	

%


\condinc{}{

\subsubsection{Proof of \Cref{lem:transformation}  }
\begin{figure}
	\centering
	\includegraphics[width=0.9\linewidth]{figures/fig_transformation}
	\caption{}
	\label{fig:transformation}
\end{figure}
\todocf{to update}
Let us revisit the construction of $\D$, one atom at  a time (as illustrated in \Cref{fig:transformation}). 
\SLV{}{
	\begin{lemma}\label{lem:I}With the same notations and assumptions as in \Cref{def:transformation}, the following properties hold
		for each $X\in \At{\M}$.

		\begin{enumerate}
			\item $X$ supports at most one conclusion   in $\Delta$ (the output of  $\D$); 
			if $\X\in \Delta$, then  $\X$ is  a conclusion of exactly one net in $\Nets_Z$.

			\item No edge in $\M'$ has label $\Xcirc$.
			\item  $\M'$ has a   single connected component for each $X\in \At{\M'}$, and it is normal.
		\end{enumerate}
	\end{lemma}
}
\begin{proof}
	Consider  $\InTrees{Z} \subseteq \Nets$.
	For each \atom $X\in \At{\Gamma_Z}$,  $\M$ contains exactly one node $\cut(\X,\Xn)$, and 
	exactely one connected component $\M(X)$, which is as in \Cref{lem:}.
	
	The positive edge $f:\X$ of the cut is conclusion  of some $\N^X\in \Nets$; 
	the negative edge $e:\Xn$  is conclusion of either    a $\w$-node, or a  $@$-node or  of a net in $\Nets$. 
	
	$\M(X)$ has  a unique positive premise (the edge $f$) and a (possibly empty) set of  negative premisses $e_i:\X_n$  ($i\in I$) where each 
	$e_i$ (if any) is conclusion of a distinct $\N_i\in \Nets$.

	%
	\begin{enumerate}[a.]
		\item Assume that each \emph{negative} pending premises of $\M(X)$ (if any) is  conclusion of some $\N\in \Nets_Z$.
		Step (ii) in \Cref{def:transformation} does not apply. 
		\begin{itemize}
			\item Case 
			$\N^X\in \Nets_Z$. 	This case includes $X=Z$ by construction (and possibly other \atoms), and the case of $e:\Xn$ root of a $\weak$-node.

			At step (iii) each  negative pending  premise  $e_i$  is $\ax/\cut$ expanded; the  conclusion $e_i':\X$ of the new $\ax$-node     is connected via the new cut  to the same $\N_i$.  	Clearly,   $X\not \in \At{\Delta}$, hence condition (1.) trivially holds. 
			Since we are assuming that no 
			$\N\in (\Nets-\Nets_Z)$ has a conclusion $\Xcirc$, then 
			$ X\not\in \At {\M'}$. Hence 
			conditions  (2.) holds.  Condition 3. is also  trivially satisfied.
			
			\item Case 
			$\N^X\in \Nets_Z$. 
		\end{itemize}

		\item Otherwise, we examine the negative edge of  $\cut(\X,\Xn)$,  which  is conclusion of either    a $\w$-node, or a  $@$-node or  of a net in $\Nets$. 
		
		\begin{itemize}
			\item Assume $e:\Xn$ is the root of a $@$-node (\Cref{fig:transformation}). If at step (ii) no other edge in $E_\D$ incides the $@$-node, then $e$ is simply $cut/ax$expanded at step (iii), and all invariants are immediate.
			
			Otherwise, at step (ii) the $@$-node is split in two (\Cref{fig:transformation}). Let us examine the $@$-node in $N_\D$. The edges which were already selected at step (i.), are all connected to some $\N\in \Nets_Z$. The edge created at step (ii) is possibly $\cut-\ax$ expanded at step (iii). Of the new edges, 
			exactly one edge $e_0$  has source but no target in $N_\D$ and therefore is a  conclusion of $\D$. Such an edge $e_0$ is labeled by either $\X$ or $\Xn$, and is a premise of $\M'$. So , $e_0$ belongs to $\Delta$ (the output of $\D$)---it is the only edge in $\Delta$ labelled by $\Xcirc$. 
			Observe moreover that if  $X$ occurs positive in $\Delta$, then necessarily $X$ occurs positive in one $\N\in \InTrees{Z}$ and therefore does not in any other $N'\in \Nets$ (by invariant (1.)). Therefore condition (1.)  hold. Conditions (2.) trivially holds (because $X\in \omega$ only concerns case A. Condition 
			(3.) is  immediate by construction.

			\item If $e:\Xn$ is a single edge,  which is conclusion of    a net $\N^e\in \Nets$, and it is connected via a cut to the positive conclusion of some  $\N^{X}\in \Nets$. If $e:\Xn$ is selected at step (i), we have that   either  $\N^e$, or $\N^{X}$ belong to $\InTrees{Z}$. In either  case, $\D$ contains an $\ax$-node of conclusions $\X,\Xn$. It is immediate  that all conditions  are satisfied.
			\item If $e:\Xn$  is conclusion of a $\w$-node, necessarily we are in case (a.). 
		\end{itemize}

	\end{enumerate}
	
\end{proof}

}
%
%

	$\R_\omega$  is a cut-net  of components $\Ax(\R_\omega) \cup \Wirings{\R_\omega}$, where $\Ax{(\R_\omega)}=\Ax(\R)$
and each $\D\in  \Wirings{\R_\omega} $ is   \emph{cut-free}. 

%% file: 99_MessagePassing.tex
\section{BN's and Proof-Nets: Inference}\label{app:Binference}
\subsection{Exact Inference over Bayesian Networks}\label{app:MP}
We briefly introduce message passing over cliques trees, and 
 we sketch  a  way to generate a clique tree $\Cliques \netB\omega$ for a \BN $ \netB $ from an elimination order $\omega$.


To do so,  it is useful to first recall  a basic  inference algorithm,   Variable Elimination  (\Cref{alg:VE}), which is simple and intuitive, and analyze its control flow.

\subsubsection{Variable Elimination}

We  recall the Variable Elimination Algorithm (\Cref{alg:VE}), which given a \BN $\netB$ over variables $\bX$, a variable $Y\in \bX$, 
and an ordering of $\bX-Y$, computes the marginal  probability $\Pr(X)$.

\newcommand{\Factors}[1]{\mathsf{Factors}^{}}

\begin{algorithm}
	\caption{Sum-Product Variables Elimination}
	\label{alg:VE}
	\begin{algorithmic}[1]

		\Require 
		\State $\netB $  \Comment{\BN over variables $\bX$}
		\State $Y \in \bX$   \Comment{A variable }
		\State $\omega$  \Comment{An elimination order  on  the {n-1} variables $\bX - {Y}$ }
		\State{$ \Phi$ }  \Comment{The set of  CPTs in  $\netB $}
		
		\Ensure  the marginal $\Pr(Y)$
		
		\Procedure{Variables Elimination}{}	
		
		\State	 $\Factors{0} \gets \Phi$  \Comment{CPTs of $\netB $ }
		
		\For{j = 1 to (n-1)}
		\State $Z\gets\omega(j)$  \Comment{The j-th variable to eliminate}
		\State $\mathsf{Factors}_Z\gets \{\phi \in \Factors{j-1} \st Z\in \scope \phi\}$
		
		\State $\psi^j \gets  \BigFProd\limits_{\phi \in \mathsf{Factors}_Z} \phi $  \Comment{Product}
		\State$\tau^j \gets \sum\limits_z \psi^j$  \Comment{Sum}
		\State  $\Factors{j} \gets  \{\tau^j\} \cup(\Factors{j-1}-\mathsf{Factors}_Z)$  
		
		\EndFor
		\Return  $ \BigFProd\limits_{\phi\in {\Factors{n-1}}}\phi  $
		
		\EndProcedure
	\end{algorithmic}
\end{algorithm}

Any elimination order will produce the correct marginal, however the cost of the computation depends on the quality of the order  $\omega$. Notice that 
the cost of running \Cref{alg:VE} is dominated by the product  produced at line 10. 
\begin{property}[Cost of VE and width]\label{fact:VEcost} Given a \BN $\netB$ and an order $\omega$, if   the largest factor $\tau_j$  produced by the \Cref{alg:VE} has $w$ variables, then 
	the time and space cost of running  the   algorithm is 
	$ \BigO{n~ \Exp{w}} $, where $ n $ is the number of variables in the Bayesian network $\netB$. 
\end{property}
The number $w$ is known as the 
width of  the order $\omega$, and  is a  measure of its  quality.

\subsubsection{The Clique Tree generated by $\omega$: analyzing the control flow of VE}\label{app:cliques}
We can graphically represent   the flow of computation determined by the  execution of  \Cref{alg:VE} with elimination order $\omega = X_1  < \dots < X_n$, as a graph. 
The  nodes are  $\{1, \dots, n\}$ (one per eliminated variable); there is an oriented edge $(i,j)$ whenever $\tau^{i}$ is used at step $j$ to compute $\tau^{j}$.
To each vertex $i$ we associate 
 the set of variables which appears in $\psi_i$, noted $\cC_{i}$, and  the set  $\Phi^{i}$ of the CPTs which are used at step $i$.
%
The above graph is a \emph{tree} because each $\tau_i$ is only used once. 

The tree together with the function which assigns $\cC_i$ to each node is a  \textbf{\emph{clique tree}} (see \Cref{def:clique_tree}), namely  a clique tree for $\netB$ induced by $\omega$,  which we denote  $ \Cliques \netB\omega $. Recall that 
its   width   (\Cref{def:clique_tree})  is defined   as the size (\ie, the number of variable) of the largest $\cC_i$ minus one. Notice that it is the same as the width of the order we started with. 
Moreover, the number of nodes is clearly at most $n$ (one node per eliminated  variable).
\condinc{}{
Formally, the tree $ \Cliques \netB \omega  $ which we have obtained  above  is a \emph{clique tree}, 
 a data structure which  allows for  efficient inference algorithms. }

Notice also that we have defined an assignment of each factors of  $\netB$ to a node in the tree. We will need this for Message Passing.
\begin{notation}Let  $\netB$ be a \BN,   $\Phi$ the set of its CPT's, and $(\tree, \cC)$ a clique tree for $\netB$.
	An \textbf{assignment} $\alpha$ of $\Phi$ to $(\tree, \cC)$  is a function $\alpha$ which associates each  $\phi^X\in \Phi$ to a node $i$ such that $\Var{\phi^X} \subseteq \cC_i$. We write $\Phi^i $ for  the set of CPTs associated to $i$.
\end{notation}

\subsubsection{Message Passing}

Given   the clique tree $\Cliques \netB \omega =(\tree,\cC)$, we can perform inference in   terms of Message Passing (\Cref{alg:messages}),
which in a way   extends \Cref{alg:VE}. 
The distinctive advantage  of the  Message Passing is that it    allows  to compute  (over  \emph{the same clique tree}) 
the marginal of \emph{every single  variable} \footnote{This, independently from  the fact that  we had obtained $\Cliques \netB \omega$ from the data flow of executing \Cref{alg:VE} with one specific   order on $\bX$, to compute a specific  marginal $Y$. }  $Z\in \bX$, simply by choosing as root $r$ a node in the tree $\tree$ such that $Z$ belongs to $\cC_r$. The cost of computing $\Pr(Z)$ will be  $ \BigO{n~ \Exp{w}}$.
Furthermore, 
by exploiting the properties of clique trees,  message passing can  compute the marginal for  \emph{all  the $n$ variables} in the original \BN, with \emph{ only  two traversals} of the  tree, at a \emph{total} cost of $ \BigO{n~ \Exp{w}} $.

%
	%
	%
	%
	%
	%
	%
	%
	%


\Cref{alg:messages} describes  
the basic message passing, which given a clique tree for $\netB$, performs a single traversal of the tree, computing  the marginal of a single (but arbitrary) variable.
%
\condinc{}{The inputs are a \BN $\netB$ over $\bX$, and a clique tree $(\tree, \cC)$ for $\netB$ together with a function which assign each \CPT of $\netB$ to a unique node.   The flow-graph $\Cliques \netB \omega$ provides  such a tree and assignment.
}
To compute $\Pr(Z)$ (for any arbitrary $Z\in \bX$), 
first, choose as root a node $r$ such that $\cC_r$ contains $Z$.
The  choice of a root determines an  orientation for the tree  (all edges are directed towards $r$).    
We can view the algorithm  as a process
of passing a message  from the leaves  towards  the root. 
Node  $i$ 
passes  a message ($\tau_{i}=  \project{\psi^i}{\cS_{ik} }$)   to its  parent  node $k$. 
When $k$ receives all the  input messages, it processes  the inputs together with its owns factors, and send the message $\tau^k$ forward.  
A node  $k$ is ready to send a message only when it has received all messages from its  children.


\begin{algorithm}
	\caption{Message Passing  (single traversal)}
	\label{alg:messages}
	\begin{algorithmic}
		\Require 
		\State $\netB $  \Comment{\BN over variables $\bX$}
		\State $Z \in \bX$   \Comment{A variable }
		\State $\Phi$ \Comment{The CPTs of $\netB$}
		\State $(\tree,\cC), \alpha$  \Comment{A   clique tree  and an assignement for  $\Phi$}
		\State{$ r$}  \Comment{root: node in the tree $\tree$ with  $Z\in  \cC_r$}

		\Ensure  the marginal $\Pr(Z)$
		
		\Procedure{Message Passing (single pass)}{}
		
		%
		%
		%
		%
		
		\State{Let  $\tree$ be   oriented, directing all edges towards  root    $r$. }

		\State{Let $I_k$ be the set of  children of node $k$ and $\cS_{ik} = \cC_i\cap \cC_k$ }
		\State	$\psi^k \leftarrow    
		\big(   \BigFProd_{\phi\in \Phi^k} \phi \big)   \FProd  \Big( \BigFProd_{i\in I_k}  \project{\psi^i}{\cS_{ik} } \Big)$ 
		
		%
		\Return $  \project{\psi^r}{Z} $
		
		\EndProcedure
	\end{algorithmic}
\end{algorithm}

\condinc{}{
	\begin{algorithm}
		\todocf{\pink{This can easily be made inductive}}
		\caption{Message Passing Algorithm (single traversal)}
		\label{alg:messages}
		\begin{algorithmic}
			\Require 
			\State $\netB $  \Comment{\BN over variables $\bX$}
			\State $Y \in \bX$   \Comment{A variable }
			\State $\Phi$ \Comment{The CPTs of $\netB$}
			\State $(\tree, \alpha)$  \Comment{A   tree and an assignement of the CPTs  }
			\State \Comment{to nodes of $\tree$ . $ \Phi^k $ denotes $\alpha^{-1}(k)$} 
			\State{$ r$}  \Comment{root: node in the tree $\tree$ with  $Y\in  \cC_r$}

			\Ensure  the marginal $\Pr(Y)$
			
			\Procedure{Message Passing (single pass)}{}
			
			\State{Let  $\tree$ be   oriented, directing all edges towards  root    $r$. }

			\While{root $r$ is not marked as ready}
			\State choose a node $k$ which is not ready and  whose inputs  $i\in I$ are all marked as ready.  Let $j$ be the target of $k$.
			
			\State	$\tau^k \leftarrow    
			\project {\big(\BigFProd_{\phi\in \Phi^k}\phi \big)\FProd   \big(\BigFProd_{i\in I} \tau^i \big)}  {\cS_{kj}} $
			\State mark $k$ as ready\\
			\EndWhile
			
			\Return $  \project{\tau^r}{Y} $
			
			\EndProcedure
		\end{algorithmic}
	\end{algorithm}
}

We recall the following standard results 
\begin{property}[Message passing and VE]\label{fact:mpVE} Let    $\netB$ be a \BN over $n$ variables $\bX$.
	\begin{enumerate}
		\item The graph $\Cliques \netB \omega$ induced by the VE algorithm is a clique tree for $\netB$.
		
		\item For every $Z\in \bX$, and every assignement of the CPT's of $\netB$ to $\Cliques \netB \omega$,   \Cref{alg:messages} computes $\Pr_\netB(Z)$ 
		with cost $\BigO{n\cdot \Exp{n}}$. 
	\end{enumerate}
	\condinc{}{Observe that if  $\omega = X_1, \dots, X_n$, the cost of computing $\Pr_\netB(X_n)$ by running \Cref{alg:VE} over $\netB$ with order $\omega-X_n$ is the same as the cost of running \Cref{alg:messages} over $\Cliques \netB \omega$.}
\end{property}
For an opportune adaptation of the message passing algorithm,  implementing a second traversal of the clique tree (from the root upwards), \emph{all the marginals} can be computed, with a \emph{total} cost which is only twice the cost of computing a single marginal.

Here we have generated a clique tree for a \BN $\netB$ starting from an elimination order $\omega$---it should be noticed that this is not the only  way. Several  algorithms are available  to \emph{compile} a BN into a clique tree. Several algorithms  are also available  to compute an elimination order. 
Poly-time algorithms are able to convert between the two, while preserving the respective width (hence, the cost of computation.
\condinc{}{It is important to stress that to compute an order or a clique tree of \emph{minimal width} is an NP-hard problem. In practice, however,  heuristics can  produces good results. }

\subsection{A proof-theoretical counter-part to message passing}

\begin{prop*}[\ref{fact:tree} Clique tree]  Let $\R$ be a   factorized \bpn, and $\tree$  the restriction of its correction graph  $\tree_\R$ to $\Wirings{\R}=\{\D_1, \dots, \D_n\}$. 
	The tree $\tree$   equipped with the function $\cC$ which  maps each $\D_i$ into  $\At{\D_i}$ is a \emph{ clique tree}.
\end{prop*}

\SLV{}{}{
	All properties in \Cref{def:clique_tree} are easy to verify.
	The jointree property is  given by the following Lemma.
	
	\begin{lemma}Let $\tree(\R)$ be as in \Cref{fact:tree}.
		If $X\in \At{\D_i}$ and $X\in \At{\D_j}$, then $X\in \At{\D_k}$, for each  $\D_k$ which is in the path between $\D_i$ and $\D_j$ in $\tree(\R)$.
	\end{lemma}
	
	\begin{proof}
		Notice that $\tree$ is a tree  which has for nodes $\Wirings{\R}=$ and an edge  $(\D_i, \D_j)$ exactly when there is a cut between $\D_i$ and $\D_j$.
		
		The claim follows from the fact that $\nax{\R}$ is  \good (\Cref{lem:B_connected}  	).
		Let $e_i$ and $e_j$ be two edges supported by $X$, with  $e_i$ in  $\D_i$  and $e_j$
		in $\D_j$.  Since $\nax{\R}$ is \good,   in $\R$ there is path $\mathfrak r$ between $e_i$ and $e_j$ whose edges are all supported by $X$. Such a path 
		is necessarily reflected into a path $\patht $ in $\T_{\R}$.

		Since  $\T_{\R}$ is a tree, there is only one path between $ \D_i $ and  $\D_j$. We prove (2.) by induction on the path $\patht$. 
		If  $i=j$, the claim is trivially true. Otherwise,  let $\patht=\D_i,\D_h, \dots, \D_j$. 
		Necessarily, in the cut-net  $\R$ there is a  $\cut(\X,\Xn)$ connecting  the net $\D_i$ and  the net $\D_h$, and so $X\in \At{\D_h}$.
		We conclude by \ih.
		
		
	\end{proof}
}

%
%
%
%

\begin{prop*}[\ref{prop:cost}]With the same assumptions as above, let $w$ be  the width of the clique tree $\Cliques {\netB} {\omega}$  induced by $\omega$, it holds:
	\begin{center}
		$\Width{\R_\omega}$ is no greater than   $w$.
	\end{center}
	
	%
\end{prop*}
By examining the construction of $\Cliques {\netB}{\omega}$ which we have sketched in \Cref{app:cliques},
one  can easily realize that $\R_\omega$ has a similar structure to that of $\Cliques {\netB}{\omega}$, and so a similar width.

		\begin{remark}The only subtlety in relating $\Cliques \netB \omega$ and $\R_\omega$, for the same $\omega$,  
			is  that    \Cref{alg:VE} has a step $i$ for each  $1\leq i\leq n$ (yielding a clique).
			However, $\R_\omega$ does not have a \wiring $\D_{X_i}$ for each $1\leq i\leq n$, because for some \atoms $ X_i $ in $\omega$,  $\Nets_{X_i}$ may be empty.
			So, for each $\D_{X_i}$ there is a step $i$ in  \Cref{alg:VE}, but the converse is not true.
		\end{remark}


		\condinc{}{
		\paragraph{Discussion}
		The cost   for the interpretation does not take into account the preliminary  factorization of $\R$. Notice however that this transformation (which is only concerned with the proof-net, not with  the quantitative information) needs only to be performed once---it is comparable to the compilation of a \BN into a clique tree. 
		
		Cost-wise, the  critical aspect in   the factorization process is  the quality of the factorization (its width), which depends on the choice of the order. Again the problem is analogous to computing a good (ideally, optimal) elimination order---we  rely on the same algorithms. Againnotice that---like for message passing on clique trees---to produce a factorized form we only need a single order. Inference of  all the marginals can be   performed on the \emph{same} factorized \bpn. 
		
		Furthermore, borrowing from message passing, we can adapt turbo interpretation in order to compute  all marginals $\sem{\R_Y}$, for all $X\in \bX$,
		with a total cost that is only twice that of computing a single marginal.
	}

\condinc{}{
\paragraph{Cost of the interpretation and comparison with VE}

	\todocf{to rewrite (simpler)}
	\begin{lemma}\label{lem:induced_clique}
		
		Let $\R_\omega$  be as in \Cref{prop:canonical}. 
		Let  $\netB = \bnet{ \R}$ be the   \BN  associated to $\R$.

		\begin{enumerate}
			\item For $1\leq i \leq n$,    $\Phi^i = \{\phi \st \Ax(\phi) \text{ is an input of } \D_i\}$
			is exactly the set of the CPTs  used by \Cref{alg:VE} at step $i$.

			\item Let $\tree(\netB,\omega)$ be the clique tree induced by  running  \Cref{alg:VE} with order $\omega$ (see \Cref{sec:VE}), and $\cC_i$ ($1\leq i \leq n$) its cliques. Then $\At{\D_i} \subseteq \cC_i$.
		\end{enumerate}
		
	\end{lemma}
	
	\begin{proof} 
		\textbf{Point 1.} Immediate by construction. 		Therefore, $\Var{\Phi^i} \subseteq \cC_i$.
		
		\textbf{Point 2.}	Assume $\D_j$ has an  input  $\Delta_i$, which is output  of $\D_i$. 
		We prove that $\Var{\Delta_i} \subseteq  \cC_i \cap \cC_j$. 
		
		Let us examine the step which produces  $\D_i$ and  the decomposition $\R_{\omega_i}=\pair {\Nets_i}{ \M_i}$.
		Every  $X\in \Var{\Delta_i}$ must appear in some input of $\D_i$, and so must appear in some $\Ax(\phi)$ which is input of a \wiring  $\D_h$ such that  $h\leq i$. Since $X\in \Var{\Delta_i}$,  and $\Delta_i$ is the output of $\D_i$,
		{$X$ must also appear in  $\M_i$ and in a generalized axiom $\Ax(\phi')$  which is   input of a \wiring  $\Delta_k$, with $k> i$.} By point (1.), $X\in \cC_h$ and $X\in \cC_k$. Observe that {if we remove edge $ij$, the node $h$ must be  connected to $i$ (possibly, $i=h$) and the node $k$ is connected to $j$ (possibly, $k=j$).}
		By the jointree property of $\tree(\netB,\omega)$, $X$ must appear in the cliques $\cC_i$ and $\cC_j$, and so in  $\cC_i \cap \cC_j$. Hence $\Delta_i \subseteq \cC_i \cap \cC_j$. 
		
		This implies that $\Var{\D_j} \subseteq \cC_j$, because 
		$\Var{\D_j} = \Var{\Ax(j)}\cup \bigcup_{i\in I} \Delta_i$, where $\Delta_i \st i\in I$ are the inputs of $\D_j$ which are conclusion of some \wiring $\D_i$. We have already proved that $\Delta_i \subseteq \cC_j $. Moreover (by point 1.) $\Var{\Phi^j}\subseteq \cC_j$.

	\end{proof}
	
}

		\begin{remark}[Discussion]
	The cost   for the interpretation does not take into account the preliminary  factorization of $\R$. 
	This is comparable to the compilation of a \BN into a clique tree. 		
	Cost-wise, the  critical aspect in   the factorization process is  the quality of the factorization (its width), which depends on the choice of the order. The problem is again the same as for  inference---we can  rely on the same algorithms and heuristics.
\SLV{}{	Finally, borrowing from message passing, we can adapt turbo interpretation in order to compute  the marginals of all 
	variables with only two traversals of the cut-net.
}
\end{remark}

%% file: main.bbl
\newcommand{\online}[1]{Available at \url{#1}}
\providecommand{\bysame}{\leavevmode\hbox to3em{\hrulefill}\thinspace}
\providecommand{\MR}{\relax\ifhmode\unskip\space\fi MR }
\providecommand{\MRhref}[2]{%
  \href{http://www.ams.org/mathscinet-getitem?mr=#1}{#2}
}
\providecommand{\href}[2]{#2}
\begin{thebibliography}{10}

\bibitem{CastellanCPW18}
Simon Castellan, Pierre Clairambault, Hugo Paquet, and Glynn Winskel, \emph{The
  concurrent game semantics of probabilistic {PCF}}, Proceedings of the 33rd
  Annual {ACM/IEEE} Symposium on Logic in Computer Science, {LICS} 2018,
  Oxford, UK, July 09-12, 2018 (Anuj Dawar and Erich Gr{\"{a}}del, eds.),
  {ACM}, 2018, pp.~215--224.

\bibitem{danosehrhard}
Vincent Danos and Thomas Ehrhard, \emph{Probabilistic coherence spaces as a
  model of higher-order probabilistic computation}, Information and Computation
  \textbf{209} (2011), no.~6, 966--991.

\bibitem{DanosH02}
Vincent Danos and Russell Harmer, \emph{Probabilistic game semantics}, {ACM}
  Trans. Comput. Log. \textbf{3} (2002), no.~3, 359--382.

\bibitem{multiplicatives}
Vincent Danos and Laurent Regnier, \emph{The structure of multiplicatives},
  Archive for Mathematical Logic \textbf{28} (1989), 181--203.

\bibitem{DarwicheHandbook}
Adnan Darwiche, \emph{Bayesian networks}, Handbook of Knowledge Representation
  (Frank van Harmelen, Vladimir Lifschitz, and Bruce~W. Porter, eds.),
  Foundations of Artificial Intelligence, vol.~3, Elsevier, 2008, pp.~467--509.

\bibitem{DarwicheBook}
Adnan Darwiche, \emph{Modeling and reasoning with bayesian networks}, Cambridge
  University Press, 2009.

\bibitem{EhrPagTas14}
Thomas Ehrhard, Michele Pagani, and Christine Tasson, \emph{Probabilistic
  {C}oherence {S}paces are {F}ully {A}bstract for {P}robabilistic {P}{C}{F}},
  The 41th Annual ACM SIGPLAN-SIGACT Symposium on Principles of Programming
  Languages, POPL14, San Diego, USA (P.~Sewell, ed.), ACM, 2014.

\bibitem{EhahrdPT18fa}
\bysame, \emph{Full abstraction for probabilistic pcf}, J. ACM \textbf{65}
  (2018), no.~4.

\bibitem{EhrhardT19}
Thomas Ehrhard and Christine Tasson, \emph{Probabilistic call by push value},
  Log. Methods Comput. Sci. \textbf{15} (2019), no.~1.

\bibitem{ll}
Jean-Yves Girard, \emph{Linear logic}, Theor. Comput. Sci. \textbf{50} (1987),
  1--102.

\bibitem{goi1}
\bysame, \emph{Geometry of interaction {I}: an interpretation of system ${F}$},
  Logic Colloquium~'88 (Ferro, Bonotto, Valentini, and Zanardo, eds.),
  North-Holland, 1988.

\bibitem{goi0}
\bysame, \emph{Towards a geometry of interaction}, Categories in Computer
  Science and Logic (Providence), AMS, 1989, Proceedings of Symposia in Pure
  Mathematics n$^\circ 92$, pp.~69--108.

\bibitem{synsem}
\bysame, \emph{Linear logic: its syntax and semantics}, Advances in Linear
  Logic (Jean-Yves Girard, Yves Lafont, and Laurent Regnier, eds.), London
  Math. Soc. Lect. Notes Ser., vol. 222, 1995, pp.~1--42.

\bibitem{locus}
\bysame, \emph{Locus solum: From the rules of logic to the logic of rules},
  Math. Struct. Comput. Sci. \textbf{11} (2001), no.~3, 301--506.

\bibitem{Girard2003}
Jean-Yves Girard, \emph{Between logic and quantic: a tract}, Linear Logic in
  Computer Science (Thomas Ehrhard, Jean-Yves Girard, Paul Ruet, and Philip
  Scott, eds.), London Math. Soc. Lect. Notes Ser., vol. 316, CUP, 2004.

\bibitem{GuerriniM01}
Stefano Guerrini and Andrea Masini, \emph{Parsing {MELL} proof nets}, Theor.
  Comput. Sci. \textbf{254} (2001), no.~1-2, 317--335.

\bibitem{JacobsZ}
Bart Jacobs and Fabio Zanasi, \emph{The logical essentials of bayesian
  reasoning}, Foundations of Probabilistic Programming, Cambridge University
  Press, 2020, pp.~295 -- 332.

\bibitem{KollerBook}
Daphne Koller and Nir Friedman, \emph{Probabilistic graphical models:
  Principles and techniques}, The MIT Press, 2009.

\bibitem{KschischangFL01}
Frank~R. Kschischang, Brendan~J. Frey, and Hans{-}Andrea Loeliger, \emph{Factor
  graphs and the sum-product algorithm}, {IEEE} Trans. Inf. Theory \textbf{47}
  (2001), no.~2, 498--519.

\bibitem{popl17}
Ugo~Dal Lago, Claudia Faggian, Beno{\^{\i}}t Valiron, and Akira Yoshimizu,
  \emph{The geometry of parallelism: classical, probabilistic, and quantum
  effects}, Proceedings of the 44th {ACM} {SIGPLAN} Symposium on Principles of
  Programming Languages, {POPL} 2017,, {ACM}, 2017, pp.~833--845.

\bibitem{LagoH19}
Ugo~Dal Lago and Naohiko Hoshino, \emph{The geometry of bayesian programming},
  34th Annual {ACM/IEEE} Symposium on Logic in Computer Science, {LICS} 2019,
  Vancouver, BC, Canada, June 24-27, 2019, {IEEE}, 2019, pp.~1--13.

\bibitem{phdlaurent}
Olivier Laurent, \emph{Etude de la polarisation en logique}, Th\`ese de
  doctorat, {U}niversit\'e {A}ix-{M}arseille~{II}, March 2002.

\bibitem{Laurent03}
\bysame, \emph{Polarized proof-nets and lambda-{\(\mathrm{\mu}\)}-calculus},
  Theor. Comput. Sci. \textbf{290} (2003), no.~1, 161--188.

\bibitem{panorama}
Paul-Andr{\'{e}} Melli{\`{e}}s, \emph{Categorical semantics of linear logic},
  Panoramas et Synth{\`{e}}ses, no.~27, pp.~1 -- 196, Soci{\'{e}}t{\'{e}}
  Math{\'{e}}matique de France, 2009.

\bibitem{NeapolitanBook}
{Richard E.} Neapolitan, \emph{Learning bayesian networks}, Prentice Hal, 2003.

\bibitem{Pearl88}
Judea Pearl, \emph{Probabilistic reasoning in intelligent systems - networks of
  plausible inference}, Morgan Kaufmann, 1988.

\bibitem{Stein2021CompositionalSF}
Dario Stein and Sam Staton, \emph{Compositional semantics for probabilistic
  programs with exact conditioning}, 2021 36th Annual ACM/IEEE Symposium on
  Logic in Computer Science (LICS) (2021), 1--13.

\end{thebibliography}
